\documentclass[11pt]{amsart}
\usepackage{amsmath, amsthm, amssymb}
\usepackage{graphicx}
\usepackage{mathbbol}
\usepackage[usenames]{color}
\usepackage{tikz}
\usepackage{alltt}
\usetikzlibrary{shapes}
\usetikzlibrary{plotmarks}

\newtheorem{theorem}{Theorem}[section]

\newtheorem{corollary}[theorem]{Corollary}   

\newtheorem{lemma}[theorem]{Lemma}   

\newtheorem{proposition}[theorem]{Proposition}

\newtheorem{definition}[theorem]{Definition}

\newtheorem{remark}[theorem]{Remark}

\newcommand{\e}{\varepsilon}
\newcommand{\grad}{\overrightarrow{\mathrm{grad}}\,}
\newcommand{\dvol}{\,\mathrm{dvol}}
\newcommand{\R}{\mathbb{R}}
\newcommand{\geucl}{g_{\mathbb{E}}}
\newcommand{\thetafinal}{\Theta_F}
\newcommand{\phifinal}{\Phi_F}
\newcommand{\omegafinal}{\Omega_F}

\newcommand{\rn}{d_{n}}

\newcommand{\mE}{m_{\mathrm{eff}}}

\begin{document}

\title[]{The effects of self-interaction on constructing relativistic point particles}

\author{Noah Benjamin}
\address{Lewis \& Clark College}
\email{noahbenjamin@lclark.edu}

\author{Iva Stavrov Allen}
\address{Lewis \& Clark College}
\email{istavrov@lclark.edu}

\date{}


\begin{abstract} We introduce a framework for studying the effects of self-interaction on the construction of point particle initial data in General Relativity. Within this framework we rigorously prove the vanishing mass claim made by Arnowitt, Deser and Misner in \cite{ZeroMass} regarding point sources. We identify a geometric structure and a scaling parameter that allow one to determine, by controlling the effects of self-interaction, when one does or does not obtain a non-zero mass.
\end{abstract}

\maketitle

\section*{Introduction}
In classical physical theories objects whose internal structure is irrelevant are commonly treated as point particles, point charges, etc. From the mathematical standpoint this is made possible by the fact that Schwarz distributions (e.g. Dirac delta distribution) are well suited for linear theories. Point particle idealization would also be useful in General Relativity, but the non-linearity of Einstein's equations makes this concept mathematically problematic. 

A notable paper addressing this issue is \cite{Strings}; in this paper metrics permitting distributional curvature are introduced and analyzed. However, the authors show that even within their wide regularity class of metrics, point-particles (sources concentrated on world-lines in space-time) are not well-defined. The authors conclude:
\begin{quote}
Indeed, it now seems likely that there is in general relativity no mathematical framework whatever for matter sources concentrated on one-dimensional surfaces in space-time.
\end{quote}

The same question is explored in the landmark 1960-62 sequence of papers by Arnowitt, Deser and Misner. Specifically, in \cite{ZeroMass} and also in \cite{ADM}, it is argued that electrically neutral point particles must have zero mass. We refer to this as the vanishing mass result. The approach taken  in \cite{ZeroMass}, as well as in this paper, deals with asymptotically Euclidean time-symmetric initial data. The Hamiltonian constraint 
$$R(g_\omega)\dvol_{g_\omega}=16 \pi \tfrac{G}{c^2} \omega$$ 
is analyzed within the conformal class of the Euclidean metric,
$$g_\omega=\theta^4 \geucl.$$
 Throughout our paper we take $\omega=\phi\dvol_{\geucl}$ to be a smooth, compactly supported matter distribution on $\R^3$ with $\phi\ge 0$. The asymptotic conditions which ensure asymptotically Euclidean data are
\begin{equation}\label{asymptotics}
\left|\partial_x^l\!\left(\theta(x)-1\right)\right|=O(|x|^{-l-1}),\ \ |x|\to \infty,\ \ l\ge0.
\end{equation}
Since $R(g_\omega)=-8\theta^{-5}\Delta_{\geucl}\theta$, the Hamiltonian constraint is equivalent to a non-linear Poisson equation 
\begin{equation}\label{TheEqn}
\theta \Delta_{\geucl}\theta \dvol_{\geucl}=-4\pi \tfrac{G}{2c^2} \omega.
\end{equation}
Observe that with the Ansatz of $\theta=1+\tfrac{G}{2c^2}V$, the approximation of \eqref{TheEqn} to first order in $\tfrac{G}{2c^2}$ simplifies to the Poisson equation,
$$\Delta_{\geucl}V=-4\pi \phi.$$
The equation \eqref{TheEqn}, paired with a boundary condition, we refer to as the Relativistic Poisson Problem (RPP). Unless otherwise stated the boundary condition is $\theta \to 1$. 

In what follows we explore the initial data obtained by taking the limit of solutions to the constraint equations corresponding to collapsing sequences of matter distributions. It should be noted that throughout the paper collapse refers not to gravitational collapse, or any dynamic process, but to the shrinking of the support of the matter distribution on each constant time slice.

In \cite{ZeroMass} the matter distribution $\omega$ is set to be a multiple $\mE\delta$ of  the Dirac delta distribution. In effect the authors argue that \eqref{TheEqn} only permits solutions when $\mE=0$, although a mathematically rigorous argument is not included. Section \ref{theta} of our paper provides such an argument. 

The reason for the vanishing mass result is, in a sense, because of interaction energies. To illustrate this we consider uncharged Brill-Lindquist  metrics (see \cite{BL})
\begin{equation}
g_{\mathrm{BL}}=\theta_{BL}^4\geucl,
\end{equation}
where $\theta_{BL}=\left(1 +\frac{G}{2c^2}\sum_{i=1}^n \frac{a_i}{|x-p_i|}\right)$ and $a_i>0$. A rough intuition behind Brill-Lindquist metrics is that they model a collection of point particles located at $x=p_i$.  Inspecting $\theta_{BL}$ suggests that one can view conformal factors as being akin to gravitational potentials. In this context, we distinguish \emph{bare mass}, $m_i$, from \emph{effective mass}, $a_i$. Brill and Lindquist attribute the discrepancy between $m_i$ and $a_i$ to \emph{interaction energy}. A direct computation shows that the asymptotic end at $x=p_i$ has ADM mass of  
$$
m_{i}=a_i\left(1+\frac{G}{2c^2}\sum_{j\neq i}\frac{a_j}{|p_i-p_j|}\right),
$$
while the asymptotic end at $x=\infty$ has the ADM mass of 
\begin{equation}\label{effectiveADM}
m=\sum a_i=\sum \frac{m_i}{\left(1+\frac{G}{2c^2}\sum_{j\neq i}\frac{a_j}{|p_i-p_j|}\right)}.
\end{equation}
In this paper it is the continuous analogue of \eqref{effectiveADM} that plays a crucial role (see \eqref{efmass}). 
 
 One can highlight the inadequacy of the framework used in \cite{ZeroMass} by inspecting the Hamiltonian constraint for the Schwarzschild body, $g_\omega=\theta_\omega^4\geucl$ where $\theta_\omega=(1+\frac{G\mE}{2c^2r})$. Recalling that $R(g_\omega)\dvol_{g_\omega}=-8\theta_\omega \Delta_{g_\omega}\theta_\omega,$ the Hamiltonian constraint for the canonical Schwarzschild geometry formally reads
$$\mE \delta+\frac{G\mE^2\delta}{2c^2r}=\omega_{\text{Schw}}.$$
We hasten to add that we are not asserting any mathematical validity of $\frac{\delta}{r}$, but it is heuristically useful for qualitative discussion nonetheless. 
This expression is surprising, as one would expect that the matter distribution corresponding to a point particle would be $\mE \delta$. However, viewing the conformal factor as analogous to a gravitational potential, the presence of $\frac{G\mE^2\delta}{2c^2r}$ suggests that the effects of gravitational self-interaction must be accounted for if one is to have non-vanishing mass. In Section \ref{NoahsSection} we provide such an account.

Our paper also presents a detailed analysis of an approximately self-similar family of distributions and the spatial geometries obtained in their limit. In this analysis we introduce a continuous parameter $\alpha$, which we interpret as determining the limiting geometry based on the degree to which one accounts for self-interaction effects. In one extreme we obtain an illuminating picture of the flat space result from \cite{ZeroMass}. In the other extreme, we obtain a Schwarzschild-type point mass. In Section \ref{NoahsSection} we make precise the connection between the parameter $\alpha$ and the effects of interaction during collapse.

\begin{figure}[h]
\centering
\begin{tikzpicture}[scale=.45]

\draw (-0.5,-6) to (6.25, -6) to (8.25,-4) to (1.5, -4) to (-0.5, -6);

\draw[thick, dashed] (2, -4.5) to [out=-30, in=100] (2.75, -6);
\draw[thick] (2.75, -6) to [out=-80, in=180] (3.75, -6.5) to [out=0, in=-100] (4.75, -6);
\draw[thick, dashed] (4.75, -6) to [out=80, in=-150] (5.75, -4.5);

\draw (8,-6) to (13.5, -6) to (15.5,-4) to (10, -4) to (8, -6);

\draw[thick, dashed] (11, -4.5) to [out=-30, in=100] (11.65, -5.25) to [out=-100, in=30] (11.25, -6);
\draw[thick] (11.25, -6) to [out=-150, in=90] (10.75, -7) to [out=-90, in=180] (11.75, -7.75) to [out=0, in=-100] (13.25, -7) to [out=80, in=-45] (13, -6);
\draw[thick, dashed] (13, -6) to [out=135, in=-90] (12.75, -5.25) to [out=80, in=-150] (13.25, -4.5);

\draw[thin] (11.65, -5.35) to [out=0, in=-150] (12.75, -5.25);

\draw (16, -5) to (19, -5);
\draw (18.8, -5.2) to (19, -5) to (18.8, -4.8);

\draw (19.5,-6) to (24.75, -6) to (26.75,-4) to (21.5, -4) to (19.5, -6);

\draw[thick, dashed] (21.9, -6) to [out=60, in=180] (23, -5.25) to [out=0, in=120] (24.1, -6);
\draw[thick] (21.9, -6) to [out=-120, in=90] (21.75, -6.5) to [out=-90, in=180] (23.25, -7.75) to [out=0, in=-90] (24.25, -6.5) to [out=90, in=-60] (24.1, -6);

\draw [fill] (23, -5.25) circle (0.1);

\end{tikzpicture}
\caption{A careful analysis of \cite{ZeroMass}; our case of $\alpha=0$.}\label{fig1}
\end{figure}
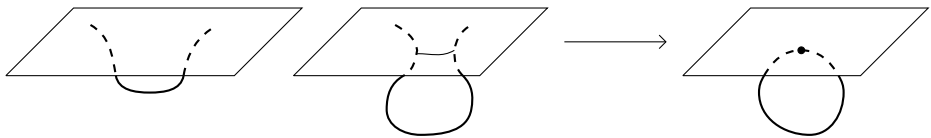

\begin{figure}[h]
\centering
\begin{tikzpicture}[scale=.45]

\draw (8,-6) to (13.5, -6) to (15.5,-4) to (10, -4) to (8, -6);

\draw[thick, dashed] (11, -4.5) to [out=-30, in=100] (11.65, -5.25) to [out=-100, in=30] (11.25, -6);
\draw[thick] (11.25, -6) to [out=-150, in=90] (10.25, -7) to [out=-90, in=180] (11.75, -7.75) to [out=0, in=-100] (13.75, -7) to [out=80, in=-45] (13, -6);
\draw[thick, dashed] (13, -6) to [out=135, in=-90] (12.75, -5.25) to [out=80, in=-150] (13.25, -4.5);

\draw[thin] (11.65, -5.35) to [out=0, in=-150] (12.75, -5.25);

\draw (16.25,-6) to (21.5, -6) to (23.5,-4) to (18.25, -4) to (16.25, -6);

\draw[thick, dashed] (19, -4.5) to [out=-30, in=100] (20.1, -5.25) to [out=-100, in=30] (19.25, -6);
\draw[thick] (19.25, -6) to [out=-150, in=90] (17.75, -7) to [out=-90, in=180] (19.75, -8.5) to [out=0, in=-100] (22.25, -7) to [out=80, in=-45] (21, -6);
\draw[thick, dashed] (21, -6) to [out=135, in=-90] (20.4, -5.25) to [out=80, in=-150] (21.25, -4.5);

\draw[thin] (20.1, -5.35) to [out=0, in=-160] (20.4, -5.3);

\draw (24, -5) to (27, -5);
\draw (26.8, -5.2) to (27, -5) to (26.8, -4.8);

\draw (27.25,-6) to (32.5, -6) to (34.5,-4) to (29.25, -4) to (27.25, -6);

\draw[thin, dashed] (30.5, -4.85) to [out=-30, in=100] (30.95, -5.25) to [out=-90, in=45] (30.75, -5.7);
\draw[thin, dashed] (31.5, -4.85) to [out=-150, in=80]  (31.05, -5.25) to [out=-90, in=135] (31.25, -5.7);

\draw [fill] (31, -5.25) circle (0.1);

\draw (27.5,-6.25) to (32.75, -6.25) to (34.75,-4.25);
\draw[dashed] (34.75,-4.25) to (29.5, -4.25) to (27.5, -6.25);

\end{tikzpicture}
\caption{The limit in the case of $0<\alpha<1$.}\label{fig2}
\end{figure}
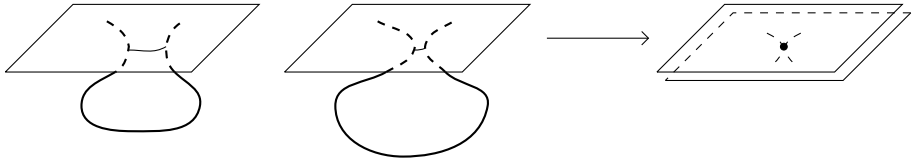

\begin{figure}[h]
\centering
\begin{tikzpicture}[scale=.45]

\draw (8.75,-6) to (14, -6) to (16,-4) to (10.75, -4) to (8.75, -6);

\draw[thick, dashed] (11.5, -4.5) to [out=-30, in=100] (12.15, -5.25) to [out=-100, in=30] (11.75, -6);
\draw[thick] (11.75, -6) to [out=-150, in=90] (10.75, -7) to [out=-90, in=180] (12.25, -8) to [out=0, in=-100] (14.25, -7) to [out=80, in=-45] (13.5, -6);
\draw[thick, dashed] (13.5, -6) to [out=135, in=-90] (13.25, -5.25) to [out=80, in=-150] (13.75, -4.5);

\draw[thin] (12.15, -5.25) to [out=-10, in=-160] (13.25, -5.25);

\draw (16.75,-6) to (21.5, -6) to (23.5,-4) to (18.75, -4) to (16.75, -6);

\draw[thick, dashed] (19, -4.5) to [out=-30, in=100] (19.8, -5.6) to [out=-100, in=30] (19.75, -6);
\draw[thick] (19.75, -6) to [out=-120, in=75] (18, -8) to [out=-90, in=180] (19.75, -9.5) to [out=0, in=-100] (22.5, -8) to [out=100, in=-75] (20.8, -6);
\draw[thick, dashed] (20.8, -6) to [out=105, in=-90] (20.7, -5.6) to [out=80, in=-150] (21.45, -4.37);

\draw[thin] (19.8, -5.8) to [out=0, in=-165] (20.7, -5.75);

\draw (23.25, -6.5) to (26, -6.5);
\draw (25.8, -6.7) to (26, -6.5) to (25.8, -6.3);

\draw (26.75,-6) to (31.5, -6) to (33.5,-4) to (28.75, -4) to (26.75, -6);

\draw[thick, dashed] (29, -4.5) to [out=-30, in=90] (29.75, -6);
\draw[thick] (29.75, -6) to [out=-90, in=25] (28.75, -8);
\draw[thick] (31.75, -8) to [out=160, in=-90] (30.6, -6);
\draw[thick, dashed] (30.6, -6) to [out=90, in=-150] (31.45, -4.37);

\draw[thick] (29.8, -6.15) to [out=-20, in=-160] (30.6, -6.15);

\draw (26.5,-9) to (31.5, -9) to (33.5,-7) to (28.5, -7) to (26.5, -9);

\end{tikzpicture}
\caption{The limit in the case of $\alpha=1$.}\label{fig3}
\end{figure}
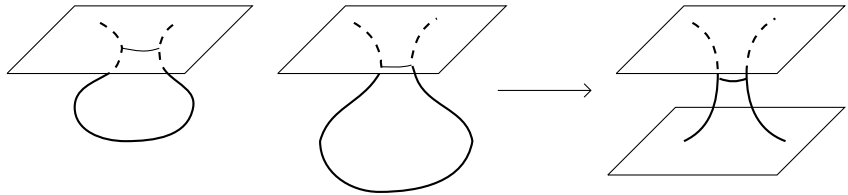

As mentioned, Section \ref{theta} is devoted to the rigorous reconstruction of the vanishing mass result. Section \ref{ssf} lays out the approximately self-similar framework and develops analytical tools for use in subsequent sections. Section \ref{NoahsSection} presents the detailed analysis of the parameter $\alpha$ as well as its interpretation. Finally, Section \ref{horizon} presents some connections to the current literature, and Appendices \ref{minsurf-ivaisms}  and B, contain detailed analysis of the event horizon which we deemed too technical to include in the body of the paper. Section \ref{horizon} and the appendices are authored by Iva Stavrov independently.

\subsection*{Acknowledgments}
This research is funded by John S. Rogers Science Research Program at Lewis \& Clark College. We thank Toby Aldape, Paul T. Allen, Mohamed Anber and Christina Sormani for useful conversations.

\section{Reconstructing the vanishing mass result}\label{theta}

Our reconstruction of the vanishing mass result takes the following course. First we prove existence and uniqueness to the RPP under the appropriate boundary condition, then we establish convergent subsequences in Sobolev spaces and finally invoke a standard argument to obtain full sequential convergence to the Euclidean metric.

\begin{proposition}\label{NoahThm}
Suppose that $\omega=\phi\dvol_{\geucl}$ is a smooth, compactly supported distribution on $\R^3$ with $\phi\ge 0$. Then there exists a unique positive solution of \eqref{TheEqn} satisfying the asymptotic conditions \eqref{asymptotics}.  
\end{proposition}

\begin{proof}
We employ a slight modification of the standard method of sub- and super-solutions. To prove existence consider the sequence of smooth functions $\theta_m$ defined by 
$\theta_0(x)\equiv 1$ and 
\begin{equation}\label{recursion}
\theta_{m+1}(x):=1+\frac{G}{2c^2}\int_{y\in \R^3} \frac{\omega(y) }{|x-y| \theta_m(y)}.
\end{equation}
By construction $\theta_{m+1}$ solves 
$$\Delta_{\geucl}\theta_{m+1}= -4\pi \frac{G}{2c^2}\cdot \frac{\phi}{\theta_m}.$$
Induction shows $\theta_m(x)\ge 1$ and 
$$\theta_0(x)\le \theta_2(x)\le \theta_4(x)\le ... \le \theta_5(x) \le \theta_3(x)\le \theta_1(x).$$

Fix compact sets $K\subseteq \mathrm{Int}(K')\subseteq K'$; without loss of generality we may assume that $\mathrm{supp}(\omega)\subseteq K$. The sequences $\theta_{m+1}$ and $-4\pi \tfrac{G}{2c^2}\phi\theta_{m}^{-1}$ are bounded in $L^2(K')$. By the interior elliptic regularity estimates we know that $\theta_m$ is bounded in $H^2(K)$. By Rellich Lemma and Sobolev inequality a subsequence of $\theta_m$ (and thus a subsequence of $\theta_{2m}$ or of $\theta_{2m+1}$) must be convergent in $C^0(K)$. Since $\theta_{2m}$ and $\theta_{2m+1}$ are monotone, at least one of them converges to some $\theta\in C^0(K)$. In fact, because of the recursive relationship \eqref{recursion} we know that both $\theta_{2m}$ and $\theta_{2m+1}$ converge in $C^0(K)$. We denote their limits by $\theta_-$ and $\theta_+$ respectively. Note that $\theta_-\le \theta_+$.

Since $\mathrm{supp}(\omega)\subseteq K$ the definition \eqref{recursion} implies that 
\begin{equation}\label{intrep}
\theta_{\pm}(x)=1+\frac{G}{2c^2}\int_{y\in \R^3} \frac{\omega(y) }{|x-y| \theta_{\mp}(y)}
\end{equation}
for all $x\in \R^3$. This integral representation of $\theta_{\pm}$ further implies that $\theta_{\pm}(x)$ are smooth and solve 
\begin{equation}\label{CoupledEqns}
\Delta_{\geucl}\theta_\pm=-4\pi \frac{G}{2c^2}\cdot \frac{\phi}{\theta_\mp}.
\end{equation}
In fact, the functions $\theta_{\pm}$ satisfy the asymptotic conditions \eqref{asymptotics} as can be seen by differentiating \eqref{intrep} under the integral sign. 

We now show that $\theta_-=\theta_+$. The function $\theta_--\theta_+$ is  non-positive and asymptotically equal to $0$. Suppose $\theta_-\neq \theta_+$. Then for some positive constant $k$ the function 
$$\theta_++k(\theta_--\theta_+)=k\theta_-+(1-k)\theta_+$$
achieves an interior minimum. In fact, by choosing $k$ sufficiently small we may assume that the function $k\theta_-+(1-k)\theta_+$ achieves a strictly positive interior minimum. Note that at this particular point of minimum we have $\Delta_{\geucl}(k\theta_-+(1-k)\theta_+)\ge 0$ while $-4\pi \frac{G}{2c^2}\phi\cdot \frac{k\theta_-+(1-k)\theta_+}{\theta_-\theta_+}<0$. On the other hand, \eqref{CoupledEqns} imply 
$$\Delta_{\geucl}(k\theta_-+(1-k)\theta_+) =-4\pi \tfrac{G}{2c^2}\phi\cdot  \tfrac{k\theta_-+(1-k)\theta_+}{\theta_-\theta_+}.$$
This contradiction shows that $\theta_-=\theta_+$, and proves the existence of solutions of \eqref{TheEqn}. 

To prove uniqueness we use the Strong Maximum Principle. If there were two positive solutions $\theta_1, \theta_2> 0$ satisfying \eqref{TheEqn}, their difference would satisfy 
\begin{equation}\label{NLP:uniqueness}
\Delta_{\geucl}(\theta_1-\theta_2)=
4\pi \frac{G}{2c^2}\cdot \frac{\phi}{\theta_1\theta_2}\cdot (\theta_1-\theta_2) 
\end{equation}
If we had $\theta_1-\theta_2 \neq 0$ somewhere, then -- without loss of generality -- the function $\theta_1-\theta_2$ would reach a positive internal maximum. However, since $4\pi \frac{G}{2c^2}\cdot \frac{\phi}{\theta_1\theta_2}\ge 0$ the Strong Maximum Principle implies that $\theta_1-\theta_2$ is a constant. This is a contradiction due to $\theta_1-\theta_2 \to 0$ as $|x|\to \infty$.
\end{proof}

\begin{definition}\label{DefnSeq}
In this Section a sequence of matter distributions $\omega_n=\phi_n\dvol_{\geucl}$ on a Euclidean background metric $\geucl$ are said to be collapsing if:
\begin{enumerate}
\item[A1] $\phi_{n}\ge 0$ for all $n$;
\medbreak
\item[A2] $\int_{\R^3}\omega_n=m$ for all $n$;
\medbreak
\item[A3] For all open sets $\mathcal{U}$ containing the origin there exists $N(\mathcal{U})$ such that for all $n\ge N(\mathcal{U})$ we have $\mathrm{supp}(\omega_n)\subseteq \mathcal{U}$.
 
\end{enumerate}
\end{definition}

It should be noted that under the stated conditions we have 
\begin{equation}\label{diracdelta}
\int_{\R^3} \varphi \omega_n \to m\varphi(0)
\end{equation}
for all test functions $\varphi$ on $\R^3$. Namely, fix a test function $\varphi$ and let $\e>0$. Then for some open set $\mathcal{U}\ni 0$ and all $x\in \mathcal{U}$ we have $|\varphi(x)-\varphi(0)|\le \e/m$. In particular, it follows that
$$\left|\left(\int_{\R^3} \varphi \omega_n\right) - m\varphi(0)\right|=\left|\int_{\R^3}\left(\varphi(x)-\varphi(0)\right)\omega_n\right|\le \int_{\R^3} (\e/m) \omega_n=\e$$
for all $n\ge N(\mathcal{U})$. 

Each $\omega_n$ will correspond, through the RPP, to a conformal factor $\theta_n$ on $\geucl$. We show that the sequence $\theta_n$ converges to the constant function $\theta_\infty=1$. In other words, the metrics $g_{\omega_n}$ converge to the Euclidean metric $\geucl$. In effect, this is the rigorous counterpart of the vanishing mass result in \cite{ZeroMass}.

\begin{theorem}\label{NoahTobyThm}
Consider a sequence $\omega_n$ of distributions satisfying the conditions of Definition \ref{DefnSeq}. The resulting sequence of solutions $\theta_n$ of the RPP converges to $1$ uniformly with all derivatives on all compact subsets of $\R^3\smallsetminus\{0\}$. 
\end{theorem}

In the following lemmas, we establish several important properties of the sequence of functions $\theta_n$. We then present the proof of Theorem \ref{NoahTobyThm}.

\begin{lemma}
\label{greensfctlemma}
For all $0<R<1$ there is an integer $N(R)$ with   
\begin{equation} \label{inequality}
\int_{y\in\R^3} \frac{\omega_{n}(y)}{|x-y|} \leq m\left(\frac{1}{|x|} +1\right)
\end{equation}
for all $n \geq N$ and all $x \in B(0,\tfrac{1}{R})\smallsetminus B(0,R)$.
\end{lemma}
\begin{proof}
Fix $0<R<1$ and consider the larger annular region 
$$B(0,\tfrac{2}{R^3})\smallsetminus B(0,\tfrac{R^3}{2}).$$ For $n$ sufficiently large we have 
$$\mathrm{supp}(\omega_n)\subseteq B(0,\tfrac{R^3}{2})$$
and consequently
$$\begin{aligned}
\left(\int_{y\in\R^3} \frac{\omega_{n}(y)}{|x-y|}\right) -\frac{m}{|x|}=&
\int_{y\in\R^3}\left(\frac{1}{|x-y|}-\frac{1}{|x|}\right)\omega_{n}(y)\\
\le & \int_{y\in\R^3}\frac{|y|}{(|x|-|y|)\cdot |x|}\omega_n(y)\\
\le &\int_{y\in\R^3}\frac{R^3}{R(2R-R^3)}\omega_n(y)=\frac{mR}{2-R^2}\le mR\le m. \qedhere
\end{aligned}$$
\end{proof}

\begin{proposition}
\label{proposition}
There exists a constant $C$ such that for all $0<R<1$ there is an integer $N(R)$ with   
\begin{enumerate}
\item \label{0} $1\leq \theta_n(x)\leq 1+\frac{C}{|x|}$
\item \label{1} $\theta_n(x) \leq C\left(\frac{1}{|x|^{1/2}}+1\right)$ 
\item \label{2} $|\partial_{x}\theta_n|\leq C\left(\frac{1}{|x|^{3/2}}+1\right)$ 
\end{enumerate}
for all $n \geq N$ and all $x \in B(0,\tfrac{1}{R})\smallsetminus B(0,R)$.
\end{proposition}

\begin{proof}
Throughout the proof we fix $0<R<1$ and assume that $n$ is sufficiently large so that $\mathrm{supp}(\omega_n)\subseteq B(0,R/2)$. For $y\in \mathrm{supp}(\omega_n)$ we then have $|y|\le R/2\le |x|/2$ and consequently 
\begin{equation}\label{rndeqn3}
\frac{1}{|x-y|}\le \frac{1}{|x|-|y|}\le \frac{2}{|x|}.
\end{equation}
 
Consider the representation formula 
\begin{equation}\label{intrep2}
\theta_n(x)=1+\frac{G}{2c^2}\int_{y\in \R^3} \frac{\omega_n(y)}{|x-y| \theta_{n}(y)}.
\end{equation}
We have 
$$\begin{aligned}
0\le \theta_n(x)-1\le &\frac{G}{2c^2} \int_{y\in \R^3}\frac{\omega_n(y)}{|x-y|}\\
\le &\frac{G}{2c^2}\cdot \frac{2}{|x|}\cdot  \int_{y\in \R^3}\omega_n(y) = \left(2m\cdot \frac{G}{2c^2}\right) \frac{1}{|x|}.
\end{aligned}$$
This completes the proof of the first of our claims. 

To prove the next claim we consider $\Delta_{\geucl}(\theta_{n}^2)$:
\begin{equation} \nonumber
\Delta_{\geucl}(\theta_{n}^2) = 2\theta_{n} \Delta_{\geucl} \theta_n + 2|d\theta_n|^2_{\geucl} \geq -8\pi \tfrac{G}{2c^2} \phi_n.
\end{equation}
Since each $\theta_{n}$ satisfies asymptotic conditions \eqref{asymptotics}, so does $\theta_n^2$. In particular, Green's representation formula applies to $\theta_n^2$ and we have that
\begin{equation}
\nonumber
\theta_{n}^2(x) =1-\frac{1}{4\pi}\int_{y\in\R^3}  \frac{1}{|x-y|}\Delta_{\geucl}(\theta_{n}^2)(y)\,\dvol_{\geucl}
\leq 1+\frac{G}{c^2} \int_{y\in \R^3} \frac{\omega_n(y)}{|x-y|}.
\end{equation}
Property \eqref{1} is now a direct consequence of Lemma \ref{greensfctlemma}.

Differentiation of \eqref{intrep2} under the integral sign yields
\begin{equation}
\nonumber
\partial_x\theta_n(x)=\frac{G}{2c^2}\int_{y\in \R^3}\partial_{x}\left(\frac{1}{|x-y|}\right) \frac{\omega_n(y)}{\theta_n(y)}.
\end{equation}
Given our assumption on $n$ we have  
\begin{equation}\label{rndeqn1}
\partial_x \theta_n(x)  = \frac{G}{2c^2} \int_{y\in B(0,R/2)}\partial_{x}\left(\frac{1}{|x-y|}\right)\frac{\omega_n(y)}{\theta_n(y)}.
\end{equation}
A direct computation shows that $$\left|\partial_{x}\left(\frac{1}{|x-y|}\right)\right| \leq \frac{1}{|x-y|^2}.$$ 
Since in our case $\frac{1}{|x-y|}\le \frac{2}{|x|}$ (see \eqref{rndeqn3})
we see that 
$$\left|\partial_{x}\left(\frac{1}{|x-y|}\right)\right| \leq \frac{2}{|x|\cdot |x-y|}.$$
Combining with \eqref{rndeqn1} produces 
$$\begin{aligned}
\left|\partial_x \theta_n(x)\right|\le&\frac{G}{2c^2}\cdot \frac{2}{|x|}\cdot  \int_{y\in B(0,R/2)}\frac{\omega_n(y)}{|x-y|\theta_n(y)}\\
=&\frac{G}{2c^2}\cdot \frac{2}{|x|}\cdot \int_{y\in \R^3}\frac{\omega_n(y)}{|x-y|\theta_n(y)}\\
=&\frac{2}{|x|} (\theta_n(x)-1).
\end{aligned}$$
The claim \eqref{2} of our Proposition is now an immediate consequence of the claim \eqref{1}. 
\end{proof}

The following is the main ingredient in the proof of Theorem \ref{NoahTobyThm}. 

\begin{proposition}\label{babyNoahTobyThm}
Consider a sequence of collapsing matter distributions $\omega_n$. The resulting sequence of solutions $\theta_n$ of the RPP has a subsequence which converges to $1$ in $H^2(K)$ for all compact subsets $K$ of $\R^3\smallsetminus\{0\}$. 
\end{proposition}

 The proof of Proposition \ref{babyNoahTobyThm} consists first of establishing convergence of a subsequence of $\theta_n$ to some $\theta_{\infty}$ on all compact $K \subseteq \R^3\smallsetminus \{0\}$, and then showing that $\theta_{\infty}=1$.

\begin{proof}
Fix a chain of compact subsets 
$$K_{0} \subseteq \mathrm{Int}(K_0')\subseteq K_{0}' \subseteq \mathrm{Int}(K_1)\subseteq K_{1} \subseteq \mathrm{Int}(K_1')\subseteq K_{1}' \subseteq ... \subseteq \R^3 \smallsetminus \{0\}$$
with $\bigcup_i K_i=\R^3\smallsetminus\{0\}$.  
Interior elliptic regularity gives
\begin{equation}\nonumber
\|\theta_n\|_{H^2(K_0)} \lesssim \|\phi_n\theta_{n}^{-1}\|_{L^2(K_{0}')}+\|\theta_n\|_{L^2(K_{0}')}.
\end{equation}
In fact, for $n$ sufficiently large so that $\mathrm{supp}(\omega_n)\cap K_0'=\emptyset$ we have 
\begin{equation*}
\|\theta_n\|_{H^2(K_{0})} \lesssim \|\theta_n\|_{L^2(K_{0}')}.
\end{equation*}
Boundedness of $\theta_n$ on all compact sets (see part \eqref{1} of Proposition \ref{proposition}) shows boundedness of $\theta_n$ in $H^2(K_{0})$. Rellich Lemma ensures a convergent subsequence of elements $\theta_{n,0}$ in $H^2(K_{0})$. An inductive argument allows us to construct a subsequence of elements $\theta_{n,i}$ of $\theta_{n,i-1}$ such that $\theta_{n,i}$ converges in $H^2(K_{i})$ for all $i \geq 1$. Since $H^2(K_{i}) \subset C^0(K_{i})$ we have that 
$$\displaystyle{\lim_{n \to \infty}} \theta_{n,i}=\lim_{n \to \infty} \theta_{n,i-1}=\theta_{\infty}.$$

Consider the sequence $\theta_{n,n}$. For compact sets $K\subseteq \R^3\smallsetminus \{0\}$  there is some $i$ with $K \subseteq K_{i}$. As a subsequence of $\theta_{n,i}$ the sequence $\theta_{n,n}$ converges to $\theta_\infty$ in $H^2(K_i)$. Consequently, $\theta_{n,n}$ converges to $\theta_\infty$ in $H^2(K)$. We conclude that we have a subsequence $\theta_{n,n}$ of elements of $\theta_n$ which converges to $\theta_{\infty}$ in $C^0(K)$ for all compact $K \subset \R^3 \smallsetminus \{0\}$. 

Furthermore, it follows from Proposition \ref{proposition} that for some constant C independent of K we have 
\begin{equation}
\label{theta infinity bound}
1 \leq \theta_{\infty}(x) \leq 1+\frac{C}{|x|} \text{\ \ and\ \ } \theta_\infty(x)\leq C\left(\frac{1}{|x|^{1/2}}+1\right),
\end{equation}
on each compact $K \subset \R^3 \smallsetminus \{0\}$. Thus, estimates \eqref{theta infinity bound} applies on all of $\R^3\setminus \{0\}$.

In what follows we prove that $\theta_{\infty}=1$ by showing that
\begin{equation}\label{WeakEqn}
\int_{\R^3} \Delta_{\geucl}(\varphi) \theta_{\infty} \dvol_{\geucl} = 0,
\end{equation}
where $\varphi$ is a test function on $\R^3$.

Fix a test function $\varphi$ and a value of $s>0$. Observe that $\theta_{n,n}\to \theta_\infty$ on $\mathrm{supp}(\varphi)\smallsetminus B(0,s)$ so that
\begin{equation*}\begin{aligned}
&\int_{\R^3} \Delta_{\geucl}(\varphi) \theta_{\infty}  \dvol_{\geucl}\\ 
=&\int_{|x|\le s} \Delta_{\geucl}(\varphi) \theta_\infty \dvol_{\geucl} + \lim_{n \to \infty} \int_{|x|\ge s} \Delta_{\geucl}(\varphi) \theta_{n,n} \dvol_{\geucl}.
\end{aligned}\end{equation*}
By \eqref{theta infinity bound} we have $\theta_{\infty}(x)=\mathcal{O}(\frac{1}{|x|^{1/2}})$ as $|x|\to 0$ and thus
\begin{equation}\label{rndeqn2}
 \int_{|x|\le s}\Delta_{\geucl}(\varphi) \theta_{\infty}  \dvol_{\geucl} = \mathcal{O}(s^{5/2}) \text{\ \ as\ \ } s\to 0.
\end{equation}
Integrating by parts twice converts the integral over $|x|\ge s$  into 
\begin{equation}
\nonumber
\int_{|x|\ge s} \varphi \Delta_{\geucl}(\theta_{n,n}) \dvol_{\geucl} \pm \int_{|x|=s} \varphi \, \grad(\theta_{n,n}) \cdot \vec{\mathrm{dA}} \pm \int_{|x|=s} \theta_{n,n} \, \grad(\varphi) \cdot \vec{\mathrm{dA}}.
\end{equation} 
Furthermore, since 
\begin{equation}\label{rndeqn4}
\Delta_{\geucl}(\theta_n)=-4\pi \tfrac{G}{2c^2}\cdot \tfrac{\phi_n}{\theta_n}=0
\end{equation}
on $|x|\ge s$ for $n$ is sufficiently large, we see that 
$$\lim_{n\to \infty} \int_{|x|\ge s} \varphi\Delta_{\geucl}(\theta_{n,n})\dvol_{\geucl}=0.$$
By Proposition \ref{proposition} we have that on $\{|x|=s\}$ $$\grad(\theta_n) \cdot \vec{\mathrm{dA}} =\mathcal{O}(s^{1/2}) \,\ \  \text{and} \, \ \  \theta_n \vec{\mathrm{dA}}=\mathcal{O}(s^{3/2})$$ 
as $s\to 0$. Overall, we obtain 
$$\lim_{n \to \infty} \int_{|x|\ge s} \Delta_{\geucl}(\varphi) \theta_{n,n} \dvol_{\geucl}=\mathcal{O}(s^{1/2}).$$
Combining with \eqref{rndeqn2} produces 
\begin{equation*}
\int_{\R^3} \theta_{\infty} \Delta_{\geucl}(\varphi) \dvol_{\geucl} =  \mathcal{O}(s^{1/2}) \text{\ \ as\ \ } s \to 0.
\end{equation*}
Taking the limit as $s\to 0$ proves \eqref{WeakEqn}.

It follows that $\theta_\infty$ is a weak -- and consequently strong -- solution of $\Delta_{\geucl}\theta_\infty=0$. By \eqref{theta infinity bound} we know that $\theta_\infty(x)\to 1$ as $|x|\to \infty$. Consequently, we must have $\theta_\infty=1$. 
\end{proof}

We are finally ready to prove Theorem \ref{NoahTobyThm}. 

\begin{proof}[Proof of Theorem \ref{NoahTobyThm}]
It remains to prove that for all compact $K\subset \R^3\smallsetminus \{0\}$ and all $k\ge 0$ we have convergence in $H^k(K)$ of the full sequence $\theta_n$ to $\theta_\infty=1$. This is done inductively on $k$, with the base case being $k=2$. 

To address the base case we suppose the opposite: that there exists some $\e_0>0$ such that for all $k\in \mathbb{N}$ there is some $n_k\ge k$ with 
\begin{equation}\label{rndeqn5}
\|\theta_{n_k}-1\|_{H^2(K)}\ge \e_0.
\end{equation}
Consider the sequence $\{\omega_{n_k}\}_{k\in \mathbb{N}}$ and the resulting functions $\{\theta_{n_k}\}_{k\in \mathbb{N}}$. By Proposition \ref{babyNoahTobyThm} we know there is a subsequence of $\theta_{n_k}$ which converges to $\theta_\infty=1$. However, this contradicts \eqref{rndeqn5}. 

Suppose that for some $k\ge 2$ and all compact $K\subset \R^3\smallsetminus \{0\}$ we have the convergence in $H^k(K)$ of the full sequence $\theta_n$ towards $\theta_\infty=1$. Now fix a compact subset $K\subset \R^3\smallsetminus\{0\}$, and let $K'\subset \R^3\smallsetminus\{0\}$ be compact with $K\subseteq \mathrm{Int}(K')\subseteq K'$. Note that for $n$ large enough the equality \eqref{rndeqn4} holds on $K'$ so that 
$$\|\theta_n-1\|_{H^{k+2}(K)}\lesssim \|\theta_n-1\|_{H^k(K')}.$$
Since by the inductive hypothesis $\theta_n\to 1$ in $H^k(K')$, we see that $\theta_n\to 1$ in $H^{k+2}(K)$. This completes our inductive proof.
\end{proof}

While the vanishing mass claim is now rigorously established, the fact that the prescribed matter should have no gravitational effect in the limit remains a troubling observation, and ultimately suggests that we do not yet have a complete understanding of the situation.

In the next section we restrict attention to approximately self-similar distributions, defined in Definition \ref{framework}, to produce a more revealing and detailed analysis. 

\section{The approximately self-similar framework}\label{ssf}
Here we set up the approximately self-similar framework and provide an example to illustrate its importance.
As before we let $\omega_n=\phi_n\dvol_{\geucl}$ with $\phi_n\ge 0$ be a sequence of distributions, this time supported on $B_{\geucl}(0,r_n)$ where $r_n\to 0$. We denote by $\Omega_n$ the sequences of approximately self-similar distributions, by $\Omega_F$ their limit, and by $\Theta_n$ and $\Theta_F$ the conformal factors arising from solving the RPP with the natural boundary condition. The details are presented in Propositions \ref{framework} through \ref{thetaexpands}.
We define self-similarity as follows.

\begin{definition}\label{framework}
Let $\omegafinal$ be a distribution on $\R^3$ and  $\rn=\left(\tfrac{G}{2c^2}\cdot\tfrac{m}{r_n}\right)^{-1}$, where $m := \int_{\R^3}^{}\omegafinal$. A sequence of distributions $\omega_n$ is said to collapse to a point of $\omegafinal$-type at the rate of $\alpha\ge 0$ if for the dilation $\mathcal{H}_{\rn}: x\mapsto \rn\,x$ and the sequence of distributions $\Omega_n$ defined by 
$$\Omega_n=(\rn)^\alpha\cdot \mathcal{H}_{\rn}^*\omega_n,$$
we have 
$$\Omega_n \to \omegafinal$$
with latter being uniform on compacts with all the derivatives. 
\end{definition}

As the example of Brill-Lindquist metrics (see the Introduction) suggests, the interaction effects blowing up as the support of $\Omega_n$ goes to zero is responsible for the ADM mass vanishing in the limit. The quantity $d_n$ contains the information about the type of point that the distributions are collapsing to, $\Omega_F$, and the rate of the collapse $r_n$. This invites the interpretation that $d_n$ is a scaling factor designed to counteract the effects of interaction for a specific collapse, captured by the conformal factor $\Theta_F$. Inspired by \eqref{effectiveADM} we make the definition 
\begin{equation}\label{efmass}
\mE=\int_{\R^3}^{} \frac{\Omega_F}{\Theta_F}
\end{equation}
The following example illustrates how the choice of $\Omega_F$ impacts $\mE$, making it clear that an approximately self-similar framework is necessary.

We consider a distribution $\Omega_{F_1}$ defined on $D=B_{\geucl}(0,1)$ and the dilation $\mathcal{H}:x \mapsto \frac{x}{2}$. We also have the distribution $\Omega_{F_2}=\mathcal{H}^*\Omega_{F_1}$ which is defined on $\mathcal{H}^*D=B_{\geucl}(0,2)$. Pulling the RPP back under $\mathcal{H}$ yields $$(\mathcal{H}^*\Theta_{F_1} ) 4\Delta_{\geucl}(\mathcal{H}^*\Theta_{F_1})  \frac{1}{8}\dvol\geucl=-4\pi \tfrac{G}{2c^2} \Omega_{F_2},$$ from which we infer $\Theta_{F_2}=\sqrt{2}\mathcal{H}^*\Theta_{F_1}$. We compute
$$
m_{\mathrm{eff,}2}= \int_{\R^3}^{} \frac{\Omega_{F_2}}{\Theta_{F_2}}
=\int_{\R^3}^{}\frac{\mathcal{H}^*\Omega_{F_1}}{\sqrt{2}\mathcal{H}^*\Theta_{F_1}}=\frac{1}{\sqrt{2}}m_{\mathrm{eff,}1}. 
$$
Once we prove Corollary \ref{adm1}, recalling that the scaling parameter $d_n^{ \ \alpha}$ is proportional to $r_n^{ \ \alpha}$, we obtain that the ratio of the ADM masses corresponding to collapsing sequences arising from $\Omega_{F_1}$ and $\Omega_{F_2}$  satisfies
$$\frac{m_{\mathrm{ADM,}1}}{m_{\mathrm{ADM,}2}} \to \frac{1}{2^{1-\frac{\alpha}{2}}}.$$ Thus the ADM mass decreases as the support of the distribution gets smaller.

 We proceed to  build on Proposition \ref{NoahThm}, establishing existence and uniqueness for the RPP in the self-similar framework with the appropriate boundary condition.  

In the analysis of $\alpha$ the series expansions for the sequences of conformal factors $\theta_n$ and $\Theta_n$ as well as $\Theta_F$ are critical. To avoid disrupting the narrative in Section \ref{NoahsSection} we develop these expansions now. We first show that $\Theta_n$ exists and expands, then that $\theta_n$ expands, and finally that $\Theta_n$ converges and $\Theta_F$ expands.

\begin{proposition}\label{Theta}
Suppose that $a_n$ is a sequence of positive real numbers and that $\Omega_n=\Phi_n\dvol_{\geucl}$ with $\Phi_n\ge 0$ is a  sequence of distributions supported in a compact domain $D$. Then there exists a corresponding sequence $\Theta_n$ of solutions of the boundary value problems
$$\Theta_n\Delta_{\geucl}\Theta_n \dvol_{\geucl}=-4\pi \tfrac{G}{2c^2}\, \Omega_n,\ \ \text{with}\ \ \lim_{|x|\to \infty}\Theta_n= a_n.$$
Furthermore, there exists a constant $C$ such that for all $x$ satisfying $|x|\ge C \mathrm{diam}(D)$ and for all $n$ we have 
\begin{equation}\label{asymptotics-for-n}
\Theta_n(x)=a_n+\frac{1}{|x|}\frac{G}{2c^2}\int_{y}\frac{\Omega_n(y)}{\Theta_n(y)}+\sum_{l=1}^\infty  \frac{G}{2c^2}\frac{C_l(x)}{|x|^{2l+1}} \int_{y} P_l(y)\,\frac{\Omega_n(y)}{\Theta_n(y)},
\end{equation}
where $C_l$ and $P_l$ are universal homogenous polynomials of degree $l$.
\end{proposition}

\begin{proof}
Considering the solutions of
$$\Theta_n \Delta_{\geucl}\Theta_n \dvol_{\geucl}=-4\pi \tfrac{G}{2c^2} \cdot a_n^{-2}\Omega_n$$
given by Theorem \ref{NoahThm} and the function $a_n \Theta_n$ gives a straightforward proof of existence of $\Theta_n$.
 We now address the series expansion of $\Theta_n$. To do se we recall that $a_n$ is the boundary condition and thus we have the Green's representation formula 

$$\Theta_n(x)=a_n+\frac{G}{2c^2}\int_{y\in D} \frac{\Omega_n(y)}{|x-y|\Theta_n(y)}.$$
Consider $y \in D$ and $x \notin D$. Note that the expression $$\frac{1}{|x-y|}=\frac{1}{|x|} \cdot \frac{1}{\left|\frac{x}{|x|}-\frac{y}{|x|}\right|}$$ can be written as $$\frac{1}{|x|} \cdot \frac{1}{\sqrt{1+T(x,y)}}$$ where $T(x,y)=\frac{1}{|x|^2} \left( -2\langle x,y \rangle+|y|^2 \right)$. It follows that we can expand $\frac{1}{|x-y|}$ using a power series, which converges when $$\left|-2\frac{\langle x,y \rangle}{|x|^2}+\frac{|y|^2}{|x|^2}\right|<1,$$ for instance if $\frac{|y|}{|x|}<\frac{1}{3}$. Inserting this expansion into the representation formula for $\Theta_n$ yields $$\Theta_n(x)=a_n+\frac{1}{|x|}\frac{G}{2c^2}\int_{y}\frac{\Omega_n(y)}{\Theta_n(y)}+\sum_{l=1}^\infty \frac{C_l\left(x\right)}{|x|^{2l+1}} \frac{G}{2c^2}\int_{y}\frac{P_l(y) \Omega_n(y)}{\Theta_n(y)}$$
where $C_l$ and $P_l$ are homogeneous polynomials of degree $l$.

This completes the proof of existence, uniqueness and the expansion of $\Theta_n$.
\end{proof}

We now build on this result to obtain a series expansion for $\theta_n$.

\begin{proposition}\label{thetaexpands}

For some constant $C$ and for all $|y|\geq C \cdot d_n \cdot \mathrm{diam}(D)$ we have $$\theta_n(y)=1+\frac{1}{|y|}\frac{G}{2c^2}\int_{z}^{}\frac{\Omega_n(z)}{\Theta_n(z)}d_n^{\frac{1-\alpha}{2}}+\sum_{l=1}^{\infty}\frac{C_l(y)}{|y|^{2l+1}}\frac{G}{2c^2}\int_{z} \frac{P_l(z) \Omega_n(z)}{\Theta_n(z)}d_n^{l+\frac{1-\alpha}{2}},$$
where $C_l$ and $P_l$ are universal homogenous polynomials of degree $l$.

\end{proposition}

\begin{proof}
We consider the sequence of functions $$\Theta_n(x)=d_n^{\frac{\alpha+1}{2}}\mathcal{H}_{d_n}^* \theta_n(x)=d_n^{\frac{\alpha+1}{2}} \theta_n(d_nx)$$  Noting that $\Delta_{\mathcal{H}_{\rn}^*\geucl}= \frac{1}{\rn^2} \Delta_{\geucl}$ and that $\dvol_{\mathcal{H}_{\rn}^*\geucl}=\rn^3\dvol_{\geucl}$, we see that under $\mathcal{H}_{\rn}^*$, the left hand side of the RPP becomes $$\rn \left( \mathcal{H}_{\rn}^* \theta_n \right) \Delta_{\geucl} \left( \mathcal{H}_{\rn}^* \theta_n \right)\dvol_{\geucl}.$$ All together, recalling that $\Omega_n=d_n^\alpha \left( \mathcal{H}_{d_n}^* \omega_n \right)$, we have that $\Theta_n$ satisfies
\begin{equation}
\Theta_n \Delta_{\geucl}\Theta_n \dvol_{\geucl}=-4\pi \tfrac{G}{2c^2} \cdot \Omega_n,\ \  \lim_{|x| \to \infty} \Theta_n=\rn^{\frac{\alpha+1}{2}}.
\end{equation} 
We now take $a_n=\rn^{\frac{\alpha+1}{2}}$ and apply Proposition \ref{Theta} at which point the change of variables $x \to \frac{y}{d_n}$  and division by $d_n^{\frac{\alpha+1}{2}}$ completes the proof. 
\end{proof}
The last task in this section is to establish the convergence of $\Theta_n$ and the asymptotics and expansion of its limit.
\begin{proposition}\label{Fconvergence}
The sequence $\Theta_n$ converges to $\Theta_F$ uniformly with all derivatives on all of $\R^3$, where $\Theta_F$ solves the RPP with Dirichlet boundary conditions. Furthermore there exists a constant $C$ such that for all $|x| \geq C\mathrm{diam}(\mathrm{D})$ we have 
$$\Theta_F(x)=\frac{G}{2c^2}\cdot \frac{\mE}{|x|}+\sum_{l=1}^\infty \frac{C_l(x)}{|x|^{2l+1}}\frac{G}{2c^2}\int_{\R^3} \frac{P_l(y)\Omega_F(y)}{\Theta_F(y)}
,$$ where $C_l$ and $P_l$ are universal homogenous polynomials of degree $l$.

\end{proposition} 

\begin{proof}
To prove existence of $\Theta_F$ observe that $\Theta_n^2$ satisfies 
$$\Delta_{\geucl}(\Theta_n^2)\ge -4\pi \tfrac{G}{2c^2} \Phi_n \text{\ \ and\ \ } \Theta_n^2(x)\to a_n^2 \text{\ \ as\ \ } x\to \infty.$$ Thus, by the representation formula we have 
\begin{equation}\label{bound-on-tpsi}
\Theta_n^2(x)\le a_n^2+\tfrac{G}{2c^2}\int_{\R^3}\frac{\Omega_n(y)}{|x-y|}\le a_n^2+\tfrac{G}{2c^2}\int_{D}\frac{\phifinal(y)+1}{|x-y|}\dvol_y,
\end{equation}
at least for sufficiently large $n$. It follows that the functions $\Theta_n$ are bounded in $L^\infty(\R^3)$.
Next, observe that on $D$ the functions $\Theta_n$ are actually bounded uniformly from below. To see this observe that the representation formula implies
\begin{equation}\label{psi-n-bdd}
\begin{aligned}
\Theta_n(x)&\ge \tfrac{G}{2c^2} \int_{\R^3} \frac{\Omega_n(y)}{|x-y|\Theta_n(y)}\\
&\ge \tfrac{G}{2c^2C}\int_{D} \frac{\Omega_n(y)}{|x-y|},
\end{aligned}
\end{equation}
where $C$ is the $L^\infty$-bound on $\Theta_n$. Our claim now follows from the fact that 
$$\int_{D} \frac{\Omega_n(y)}{|x-y|}\to \int_{D} \frac{\omegafinal(y)}{|x-y|}$$ and the fact that 
$\int_{D} \frac{\omegafinal(y)}{|x-y|}$ 
reaches its positive minimum $C_{\mathrm{min}}$ over $x\in D$: 
$$\Theta_n(x)\ge \frac{C_{\mathrm{min}}G}{4c^2C},\ \ n\gg 1.$$ 

Consider a sequence of compact subsets of $\R^3$ with 
$$K_1\subseteq \mathrm{Int}(K_2)\subseteq K_2\subseteq \mathrm{Int}(K_3)\subseteq ... \text{\ \ and\ \ } \cup_{i}K_i=\R^3.$$
Through (several) applications of interior elliptic regularity and Rellich Lemma we inductively generate subsequences $\Theta^{(i)}_n$ of $\Theta_n$ which 
\begin{enumerate}
\item[a)] converge in $H^4(K_i)\subseteq C^2(K_i)$ as $n\to \infty$, and 
\medbreak
\item[b)] are subsequences of the previously constructed subsequences $\Theta^{(i-1)}_n$. 
\end{enumerate}
The diagonal subsequence $\Theta^{(n)}_n$ converges uniformly with two derivatives on each compact subset of $\R^3$. The limit function $\thetafinal$ thus  solves the equation 
$$\thetafinal\Delta_{\geucl}\thetafinal=-4\pi\tfrac{G}{2c^2}\,\phifinal.$$

We now show that the whole sequence $\Theta_n$ converges to $\thetafinal$ on all compact sets $K$. To do so we suppose the opposite, that there is some $\e_0>0$ and a subsequence $\Theta_{n_k}$ with 
\begin{equation}\label{rnd17}
\|\Theta_{n_k}-\thetafinal\|_{L^\infty(K)}\ge \e_0
\end{equation} for all $k$. By applying what we have already proven to sequences $a_{n_k}$ and $\Phi_{n_k}$ we 
obtain a subsequence of $\Theta_{n_k}$ which converges to $\thetafinal$. The latter contradicts \eqref{rnd17} and proves that $\Theta_n\to \Theta_F$ in $L^\infty(K)$ for all compact $K$. A standard application of interior elliptic regularity gives
$$\|\Theta_n-\Theta_F\|_{H^{l+2}(K)}\le C\left(\left\|\tfrac{\Omega_n}{\Theta_n}-\tfrac{\Omega_F}{\Theta_F}\right\|_{H^l(K')}+\left\|\Theta_n-\Theta_F\right\|_{L^2(K')}\right).$$ This shows that the convergence over compact subsets is in fact with all  derivatives. 

Note that the asymptotic behavior of $\Theta_F$, as well as the claimed convergence of $\Theta_n$ to $\Theta_F$ with all derivatives on $\R^3$, follows from  \eqref{asymptotics-for-n} due to

$$\int_{D} P_l\,\frac{\Omega_n}{\Theta_n} \to \int_{D} P_l\,\frac{\Omega_F}{\Theta_F} \text{\ \ \ i.e.\ \ \ } \int_{\R^3} P_l\,\frac{\Omega_n}{\Theta_n} \to \int_{\R^3} P_l\,\frac{\Omega_F}{\Theta_F}.\qedhere$$
\end{proof}

As we make frequent use of the expansions developed above, it is worthwhile to collect the constant terms before we proceed. We define $$b_{0,n}:=\frac{G}{2c^2} \int_{\R^3}^{} \frac{\Omega_n}{\Theta_n} \ \  , \ b_{l,n}:=\frac{G}{2c^2} \int_{\R^3}^{} \frac{P_l \Omega_n}{\Theta_n} \ \ , \ b_l:=\frac{G}{2c^2} \int_{\R^3}^{} \frac{P_l \Omega_F}{\Theta_F}.$$
In addition to cleaning up the series expansions we also have the notational convenience that $b_{0,n}\to \frac{G\mE}{2c^2}$. 

Using these definitions the series expansions are can now be written 
$$
\begin{aligned}
& \Theta_n(x)=a_n+\frac{b_{0,n}}{|x|}+\sum_{l=1}^\infty \frac{b_{l,n}C_l\left(x\right)}{|x|^{2l+1}}, \\
& \theta_n(y)=1+\frac{b_{0,n}d_n^{\frac{1-\alpha}{2}}}{|y|}+\sum_{l=1}^{\infty}\frac{b_{l,n}C_l(y)d_n^{l+\frac{1-\alpha}{2}}}{|y|^{2l+1}}, \\
& \Theta_F(x)=\frac{G}{2c^2}\cdot \frac{\mE}{|x|}+\sum_{l=1}^\infty \frac{b_l C_l(x)}{|x|^{2l+1}}.
\end{aligned}
$$

With these expansions in hand we are ready to study the parameter $\alpha$.

\section{The analysis of the parameter $\alpha$} \label{NoahsSection}
To extract as much detail as possible we study specific subsets of $\R^3$ in contrast to the macroscopic approach taken in Section \ref{theta}. To that end, we view $$(\R^3, g_n)=(\R^3, \theta_n^4\geucl)$$ as arising from embedding model geometries into $\R^3$. We notate embeddings by lower case latin letters with the subscript $n$. Their domains are notated by the corresponding capital letter with subscript $n$, for example $\mathrm{i}_n:\mathrm{I}_n \to \R^3$. One of the focal points of the analysis is whether or not the mass vanishes in the limit for the various values of $\alpha$. We can deal with this now by truncating our series expansion to first order and reading off the ADM mass of $g_n$. 
\begin{corollary}\label{adm1}
We have 
$$m_{\mathrm{ADM}}(\theta_n^4\geucl)=\frac{2c^2}{G}b_{0,n}d_n^{\frac{1-\alpha}{2}}=d_n^{\frac{1-\alpha}{2}}\int_{}^{}\frac{\Omega_n}{\Theta_n}.$$
and consequently 
$$m_{\mathrm{ADM}}\to 
\begin{cases}
0, & \text{when}\ \ \alpha<1\\
\mE, & \text{when}\ \ \alpha=1\\
\infty, & \text{when}\ \ \alpha>1.
\end{cases}$$
\end{corollary}

 Before presenting the analysis on $\alpha$ we pause to make a note to the reader about the parameters $\e$ and $\delta$ in Figures \ref{fig4}--\ref{fig6}.  

The appearance of these parameters has to do with the claims of convergence being stated with respect to the $C^k$ norm on the model spaces. The parameter $\delta$ is introduced because our series expansion of $\theta_n$ fails as we approach $|x|=b_{0,n}d_n$. We use $\delta$ to create a buffer zone, past which we are safe to use the expansion. Essentially, once $\alpha$ and $k \in \mathbb{N}\cup\{ 0\}$ are fixed, we have an $\e > 0$ and $\delta>0$ such that the convergences are $C^k$ on the decomposition of $\R^3$ induced by this particular choice of $\e$ and $\delta$. The existence and specific choice of these parameters only appear sensible \emph{after} one proves the next set of propositions. 
With that said, we put our faith in the customary cooperation of the reader as we unapologetically present these choices now, devoid of their context. We require 
$$\frac{1-\alpha}{2}\cdot \frac{k}{k+1}<\e<\frac{1-\alpha}{2} \text{\ \ and\ \ }0<\delta<\text{min}\left\{\alpha,\frac{\alpha+1}{2(k+1)}\right \}.$$

We present the case of $0<\alpha<1$ first because it features the most complications. 

\subsection*{The case $0<\alpha<1$} \ 

Note that in this regime Corollary \ref{adm1} tells us that the ADM mass will vanish in the limit.

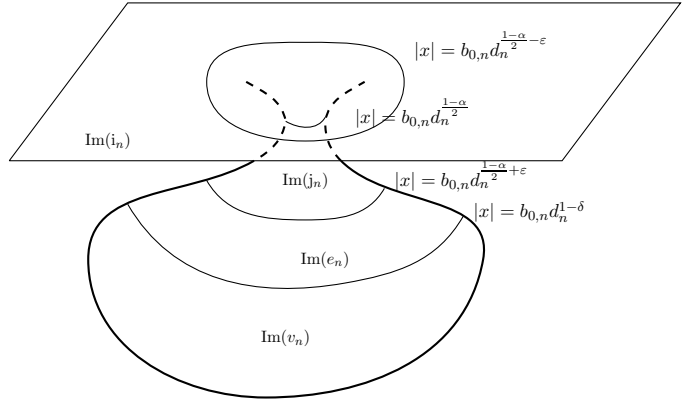
\begin{figure}[h]
\centering
\begin{tikzpicture}[scale=1.05]

\draw (16.5,-5.75) to (23.5, -5.75) to (25,-3.75) to (18, -3.75) to (16.5, -5.75);

\draw[thick, dashed] (19.5, -4.75) to [out=-30, in=100] (20, -5.25) to [out=-100, in=30] (19.6, -5.75);
\draw[thick] (19.6, -5.75) to [out=-150, in=90] (17.5, -7) to [out=-90, in=180] (19.75, -8.75) to [out=0, in=-100] (22.5, -7) to [out=80, in=-45] (20.7, -5.75);
\draw[thick, dashed] (20.7, -5.75) to [out=135, in=-90] (20.5, -5.25) to [out=80, in=-150] (21, -4.75);

\draw[thin] (20, -5.25) to [out=-30, in=-120] (20.5, -5.2);

\draw[thin] (19, -4.75) to [out=-90, in=180] (20.25, -5.5) to [out=0, in=-90] (21.5, -4.75) to [out=90, in=0] (20.25, -4.25) to [out=180, in=90] (19, -4.75); 

\draw[thin] (19, -6) to [out=-60, in=180] (20.25, -6.5) to [out=0, in=-120] (21.25, -6.1);  

\draw[thin] (18, -6.3) to [out=-60, in=-170] (20.75, -7.3) to [out=10, in=-120] (22.25, -6.45);

\node[right] at (21.5, -4.3) {\scalebox{0.65}{$|x|=b_{0,n}d_n^{\frac{1-\alpha}{2}-\e}$}};

\node[right] at (20.75, -5.15) {\scalebox{0.65}{$|x|=b_{0,n}d_n^{\frac{1-\alpha}{2}}$}};

\node[right] at (21.25, -5.95) {\scalebox{0.65}{$|x|=b_{0,n}d_n^{\frac{1-\alpha}{2}+\e}$}};

\node[right] at (22.25, -6.4) {\scalebox{0.65}{$|x|=b_{0,n}d_n^{1-\delta}$}};

\node at (20, -8){\scalebox{0.57}{$\mathrm{Im}(v_n)$}};

\node at (20.5, -7){\scalebox{0.57}{$\mathrm{Im}(e_n)$}};

\node at (20.25, -6){\scalebox{0.57}{$\mathrm{Im}(\mathrm{j}_n)$}};

\node at (17.75, -5.5){\scalebox{0.57}{$\mathrm{Im}(\mathrm{i}_n)$}};

\end{tikzpicture}
\caption{Regions addressed in Propositions \ref{In} -- \ref{Vn}.}\label{fig4}
\end{figure}
Here the embeddings and domains are

$$\begin{aligned}
& \mathrm{I_n}=\{|y| \geq b_{0,n}d_n^{\frac{1-\alpha}{2}-\e}\} &&\mathrm{i}_n:y \mapsto y  \\
& \mathrm{J_n}=\{b_{0,n}d_n^{\e} \leq |y| \leq b_{0,n}d_n^{-\e}\} &&\mathrm{j}_n:y \mapsto d_n^{\frac{1-\alpha}{2}}y \\
& \mathrm{E_n}=\{b_{0,n}d_n^{\frac{1-\alpha}{2}-\e} \leq |y| \leq b_{0,n}d_n^{\delta-\alpha}\} &&\mathrm{e}_n:y \mapsto \frac{b_{0,n}^2d_n^{1-\alpha}y}{|y|^2} \\
& \mathrm{V_n}=\{|y| \leq b_{0,n}d_n^{-\delta}\} &&\mathrm{v}_n:y \mapsto d_ny. \\
\end{aligned}$$
A direct computation shows that $\mathrm{Im}(\mathrm{I_n})\cup\mathrm{Im}(\mathrm{J_n})\cup\mathrm{Im}(\mathrm{E_n})\cup\mathrm{Im}(\mathrm{V_n})=\R^3$.

\begin{proposition}\label{In}
On $\mathrm{I_n}$ we have $$\left\|\mathrm{i}_n^*g_n-\geucl\right\|_{C^k(\mathrm{I_n},\geucl)} \to 0.$$ 
\end{proposition}

\begin{proof}
Recalling that $\mathrm{i_n}=\mathrm{Id}$, we have $|y|=|x| \geq b_{0,n}d_n^{\frac{1-\alpha}{2}-\e} \geq Cb_{0,n}d_n$, the series expansion of $\theta_n$ from Proposition \ref{thetaexpands} is valid. For all $y \in \mathrm{I}_n$ we have $$\left\|\mathrm{i}_n^*\theta_n-1 \right\| \leq \left(d_n^\e+\sum_{l=0}^{\infty}h\cdot \frac{b_{l,n}d_n^{l(\frac{1+\alpha}{2}+\e)}}{b_{0,n}^{l+1}}\right),$$ where $h$ is a universal constant. Taking the limit of the upper bound yields $\left\|\mathrm{i}_n^*\theta_n-1 \right\|_{C^k(\mathrm{I_n},\geucl)} \to 0$ for $k=0$. Computing the $k^{\text{th}}$ derivative of $\theta_n$ from the expansion yields
 $$
|\partial^k\theta_n| \leq \sum_{l=0}^{\infty}\frac{b_{l,n}d_n^{l+\frac{1-\alpha}{2}}}{|y|^{l+k+1}}=\mathcal{O}(d_n^{ \ l(\frac{1+\alpha}{2}+\e)+\e(k+1)-k\frac{1-\alpha}{2}}).
$$
Our choice of $\e$ ensures that the exponent on $d_n$ is positive and thus that $|\partial^k(\mathrm{i_n}^*g_n)| \to 0$ for $k\geq1$, which completes the proof.
\end{proof}
\begin{proposition}\label{Jn} 
On $\mathrm{J}_n$ we have 
$$\|d_n^{-(1-\alpha)}\mathrm{j}_{n}^*g_n - 
g_{\mathrm{Schw}}\|_{C^k(\mathrm{J_n}, g_{\mathrm{Schw}})}\to 0,$$ 
where $g_{\mathrm{Schw}}=(1+\tfrac{G\mE}{2c^2|y|})^4\geucl$.
\end{proposition}

  \begin{proof}
Our strategy in this proof is somewhat different and so we take a moment to discuss it qualitatively, as a similar situation arises in Proposition \ref{Vn}. Our region is composed of two sections, above the horizon, where the analysis is straightforward due to the Schwarzschild geometry being relatively close to Euclidean, and below the horizon, where the analysis is complicated by taking norms with respect to $g_{\text{Schw}}$. We will exploit the symmetry about the event horizon, given by $\Lambda_h:z \mapsto \frac{b_{0,n}^2z}{|z|^2}$, to transport the region below the horizon into territory where the analysis is easier. Our proof thus relies on the fact that on both $\mathrm{Im}(\mathrm{J_n})$ and $\mathrm{Im}\left(\Lambda_h(\mathrm{J_n})\right)$, we have $|x| \geq Cb_{0,n}d_n$, and therefore the expansion of $\theta_n$ given in Proposition \ref{thetaexpands} is applicable. First, focusing on the region $$\{b_{0,n} \leq |y| \leq b_{0,n}d_n^{-\e}\},$$ we compute $d_n^{-(1-\alpha)} \mathrm{j}_n^* \left( \theta_n^4 \geucl \right)=\theta_n(d_n^{\frac{1-\alpha}{2}}y)^4\geucl$ via the series expansion to obtain $$\left(1+\frac{b_{0,n}}{|y|}+\sum_{l=1}^{\infty}\frac{\left(d_n^{\frac{1+\alpha}{2}}\right)^lC_l(y)b_{l,n}}{|y|^{2l+1}} \right)^4 \geucl.$$ We have that the sum over $l$, and all of its derivatives, can be written as $\mathcal{O}\left( d_n^{\frac{1+\alpha}{2}}\right)$. Furthermore, boundedness of the Christoffel symbols arising from the Schwarzschild metric implies the existence of some constant $B$ such that all together    \begin{multline*}\left\| d_n^{-(1-\alpha)}\mathrm{j}_{n}^*g_n -g_{\mathrm{Sch}} \right\|_{C^k(\mathrm{J_n},g_{\text{Schw}})}
  \\ \leq B \cdot \left\| \left(1+\frac{b_{0,n}}{|y|}+\mathcal{O}(d_n^{\frac{1+\alpha}{2}})\right)^4 \geucl-\left( 1+\frac{G\mE}{2c^2|y|}\right)^4\geucl \right\|_{C^k(\mathrm{J_n},\geucl)}.\end{multline*} 
  Taking the limit of the right hand side, noting that $b_{0,n} \to \frac{G\mE}{2c^2}$, proves the claim on $\{b_{0,n} \leq |y| \leq b_{0,n}d_n^{-\e}\}$. We finish the proof  
  $$\Lambda_h:\{b_{0,n} \leq |z| \leq b_{0,n}d_n^{-\e}\} \to \{b_{0,n}d_n^\e \leq |y| \leq b_{0,n}\}.$$ A direct computation yields, upon simplification, 
  $$\begin{aligned}
  &\Lambda_h^*(d_n^{-(1-\alpha)}\mathrm{j}_{n}^*g_n)=\left(1+\frac{b_{0,n}}{|z|}+\mathcal{O}(d_n^{\frac{1+\alpha}{2}-\e})\right)^4\geucl \to g_{\text{Schw}},\\
 & \Lambda_h^*(g_{\text{Schw}})=\left(\frac{\frac{G\mE}{2c^2}}{b_{0,n}}+\frac{b_{0,n}}{|z|}\right)^4\geucl \to g_{\text{Schw}}.
 \end{aligned}$$
The convergence is with all derivatives as in the first case. Thus we can follow the example from the first region to complete the proof.
 \end{proof}
Before proceeding to the analysis on $\mathrm{E}_n$ we pause to draw attention to the following computation.
\begin{remark}\label{length}
We let $0\leq \alpha <1$ and compute the length, denoted $L_n$ of the region $\mathrm{Im}(\mathrm{j}_n)$. We use the metric $\mathrm{j}_n^*g_n$. Noting that this is simply the metric from Proposition \ref{Jn} scaled by $d_n^{1-\alpha}$ we have that 
$$
L_n=d_n^{\frac{1-\alpha}{2}}\int_{b_{0,n}d_n^\e}^{b_{0,n}d_n^{-\e}} \left( 1+\frac{b_{0,n}}{r}+\mathcal{O}(d_n^{\frac{1+\alpha}{2}})\right)^2 \ dr.
$$
A direct computation yields 
$$
L_n \leq \mathcal{O}(1)(d_n^{\frac{1-\alpha}{2}-\e}-d_n^{\frac{1-\alpha}{2}+\e})+\mathcal{O}(d_n^{\frac{1-\alpha}{2}})\ln(d_n^{-2\e}),
$$
which tends to zero.
Thus for $0 \leq \alpha <1$, the region $\mathrm{Im}(\mathrm{j}_n)$ vanishes in the limit.
\end{remark}
\begin{proposition}\label{En}
On $\mathrm{E}_n$ we have $$\left\|e_n^*g_n-\geucl\right\|_{C^k(E_n,\geucl)}\to 0.$$
\end{proposition}
\begin{proof}
 From $|y| \leq b_{0,n}d_n^{\delta-\alpha}$ we have $ |x| \geq b_{0,n}d_n^{1-\delta}\geq Cb_{0,n}d_n$ and thus the expansion of $\theta_n$ holds. We compute  $$e_n^*g_n=\left( 1+\frac{b_{0,n}d_n^{\frac{1-\alpha}{2}}}{|y|}+\sum_{l=0}^{\infty}\frac{C_l(y)b_{l,n}d_n^{l\alpha}}{b_{0,n}^l}\right)^4\geucl.$$ Following the same logic as in Proposition \ref{In} we have for all $y \in \mathrm{E}_n$,
 \begin{align*}
 &|e_n^*g_n-\geucl|=\mathcal{O}(d_n^{ \ \e}) \\
 &|\partial^k(e_n^*g_n)|=\mathcal{O}(d_n^{ \ l(\frac{1+\alpha}{2}+\e)+\e(k+1)-k\frac{1-\alpha}{2}}).
 \end{align*}
 Our choices of $\e$ and $\delta$ now ensure $$\left\|e_n^*g_n-\geucl\right\|_{C^k(E_n,\geucl)}\to 0,$$ completing the proof.
\end{proof}

\begin{proposition}\label{Vn}
On $\mathrm{V}_n$ we have that $$\left\|d_n^{2\alpha}v_n^*g_n-\Theta_F^4\geucl\right\|_{C^k(V_n,\Theta_F^4\geucl)} \to 0.$$
\end{proposition}
\begin{proof}
Direct computation yields 
$$d_n^{2\alpha}v_n^*g_n=\Theta_n^4\geucl,$$
and so it may be tempting to directly apply Proposition \ref{Fconvergence}. However, we run into trouble as we move away from the origin due to the metric $\Theta_F^4\geucl$ becoming less and less Euclidean. With this in mind we employ the approach of Proposition \ref{Jn}, splitting $\mathrm{V}_n$ into two regions by excising a compact ball $\{ |y| \leq Q \}$. Being compact, our choice of norm on this region is unimportant. Working with the Euclidean metric, the proof on the ball follows directly from Proposition \ref{Fconvergence}. 

To deal with the region outside the ball $\{ |y| \leq Q \}$ we specifically choose $Q > C\cdot \mathrm{diam}(D)$ so that the series expansions of $\Theta_n$ and $\Theta_F$ are both valid on the remaining annulus 
$$\mathrm{V}_n':=\{Q \leq |y| \leq b_{0,n}d_n^{-\delta} \}.$$ 
We now transport the remaining problem into terrain where we have more control by means of the inversion $\Lambda_1:z\mapsto \frac{z}{|z|^2}$. We obtain 
$$\Lambda_1^*(\Theta_n^4\geucl)=\left(\frac{d_n^{\frac{1+\alpha}{2}}}{|z|}+b_{0,n} + \sum_{l=0}^{\infty}C_l(z)b_{l,n}\right)^4\geucl,$$
and $$\Lambda_1^*(\Theta_F^4\geucl)=\left(\frac{G\mE}{2c^2}+\sum_{l=0}^{\infty}C_l(z)b_{l}\right)^4\geucl.$$
In particular, we see that $\Lambda_1^*(\Theta_F^4\geucl)$ and $\geucl$ yield equivalent norms on $\Lambda_1^{-1}(\mathrm{V}_n')=\{\frac{d_n^{\delta}}{b_{0,n}} \leq |z| \leq \frac{1}{Q} \}$.
Recall that $b_{0,n} \to \tfrac{G\mE}{2c^2}$, $b_{l,n} \to b_l$ and $\left|\partial^k\left(\tfrac{d_n^{\frac{1+\alpha}{2}}}{|z|}\right)\right|=\mathcal{O}\left(\tfrac{d_n^{\frac{1+\alpha}{2}}}{|z|^{k+1}}\right)$. At worst we have $|z|=\frac{d_n^{ \ \delta}}{b_{0,n}}$ and thus
$$\left\|\Theta_n^4\geucl-\Theta_F^4\geucl\right\|_{C^k(\mathrm{V'_n},\Theta_F^4\geucl)} =\mathcal{O}\left(d_n^{\frac{1+\alpha}{2}-(k+1)\delta}\right).$$
The claimed convergence is now a consequence of $\delta<\min\{\frac{1+\alpha}{2(k+1)}, \alpha\}$.
\end{proof}

Before moving on to the $\alpha=0$ case, we make a remark about the region $\mathrm{V_n}$. 
\begin{remark}\label{bottomsheet}
 Recalling that $d_n^{\frac{1+\alpha}{2}}\mathcal{H}_{d_n}^*\theta_n=\Theta_n$ we have that $$\mathcal{H}_{d_n^{1+\alpha}}^*(\Theta_n^4\geucl)=\Theta_n(d_n^\alpha x)^4\geucl \to \Theta_F(0)^4\geucl,$$ from which we infer the existence of a local coordinate system centered at the origin in which the metric converges to $\geucl$. While Figure \ref{fig2} depicts the collapse viewed from infinity, this result implies that from the standpoint of the center of the collapse, the limit picture would be flat Euclidean space. 
\end{remark}

\subsection*{The case $\alpha=0$} \ 

This case is analogous to the vanishing mass result in \cite{ZeroMass}.

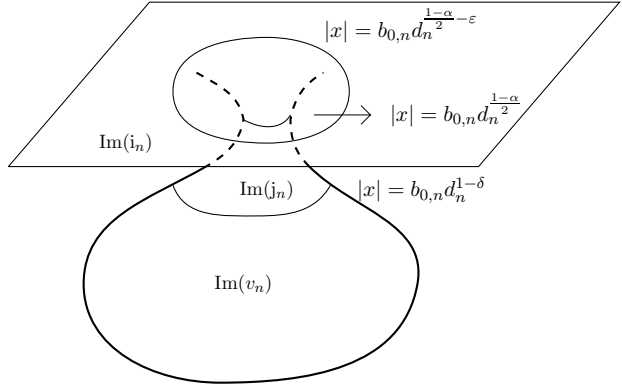
\begin{figure}[h]
\centering
\begin{tikzpicture}[scale=1.25]

\draw (17.5,-5.75) to (22.5, -5.75) to (24,-4) to (19, -4) to (17.5, -5.75);

\draw[thick, dashed] (19.5, -4.75) to [out=-30, in=100] (20, -5.25) to [out=-100, in=30] (19.6, -5.75);
\draw[thick] (19.6, -5.75) to [out=-150, in=90] (18.3, -7) to [out=-90, in=180] (19.75, -8.05) to [out=0, in=-100] (21.85, -7) to [out=80, in=-45] (20.7, -5.75);
\draw[thick, dashed] (20.7, -5.75) to [out=135, in=-90] (20.5, -5.25) to [out=80, in=-150] (20.85, -4.75);

\draw[thin] (20, -5.25) to [out=-30, in=-120] (20.5, -5.2);

\draw[thin] (19.25, -4.95) to [out=-90, in=180] (20.25, -5.5) to [out=0, in=-90] (21.125, -4.95) to [out=90, in=0] (20.25, -4.375) to [out=180, in=90] (19.25, -4.95); 

\draw[thin] (19.25, -5.95) to [out=-70, in=180] (20.15, -6.275) to [out=0, in=-120] (20.925, -5.95);

\node[right] at (20.75, -4.25) {\scalebox{0.75}{$|x|=b_{0,n}d_n^{\frac{1-\alpha}{2}-\e}$}};

\node[right] at (21.45, -5.15) {\scalebox{0.75}{$|x|=b_{0,n}d_n^{\frac{1-\alpha}{2}}$}};

\draw[ultra thin] (20.75, -5.2) to (21.35, -5.2) to (21.25, -5.1);
\draw[ultra thin] (21.35, -5.2) to (21.25, -5.3);

\node[right] at (21.1,-6) {\scalebox{0.75}{$|x|=b_{0,n}d_n^{1-\delta}$}};

\node at (20, -7){\scalebox{0.67}{$\mathrm{Im}(v_n)$}};

\node at (20.25, -6){\scalebox{0.67}{$\mathrm{Im}(\mathrm{j}_n)$}};

\node at (18.75, -5.5){\scalebox{0.67}{$\mathrm{Im}(\mathrm{i}_n)$}};

\end{tikzpicture}
\caption{Regions addressed in Proposition \ref{bubble}.}\label{fig5}
\end{figure}

In this case the embeddings and domains are
$$\begin{aligned}
& \mathrm{I_n}=\{|y| \geq b_{0,n}d_n^{\frac{1}{2}-\e}\} &&\mathrm{i}_n:y \mapsto y  \\
& \mathrm{J_n}=\{b_{0,n}d_n^{\frac{1}{2}-\delta} \leq |y| \leq b_{0,n}d_n^{-\e}\} &&\mathrm{j}_n:y \mapsto d_n^{\frac{1}{2}}y \\
& \mathrm{V_n}=\{|y| \leq b_{0,n}d_n^{-\delta}\} &&\mathrm{v}_n:y \mapsto d_ny
\end{aligned}$$
A direct computation shows that $\mathrm{Im}(\mathrm{I_n})\cup\mathrm{Im}(\mathrm{J_n})\cup\mathrm{Im}(\mathrm{V_n})=\R^3$. 

\begin{proposition}\label{bubble}
We have that
$$\begin{aligned}
&\left\|\mathrm{i}_n^*g_n-\geucl\right\|_{C^k(\mathrm{I_n},\geucl)} &&\to 0, \\
&\|d_n^{-1}\mathrm{j}_{n}^*g_n -g_{\mathrm{Schw}}\|_{C^k(\mathrm{J_n}, g_{\mathrm{Schw}})} &&\to 0, \\
&\left\|v_n^*g_n-\Theta_F^4\geucl\right\|_{C^k(V_n,\Theta_F^4\geucl)} &&\to 0.
 \end{aligned}$$
\end{proposition}
\begin{proof}
We proceed as in Propositions \ref{In}, \ref{Jn} and \ref{Vn} with the exception that the sum term in the proof of Proposition \ref{Jn} now $\mathcal{O}(d_n^\delta)$ rather than $\mathcal{O}(d_n^{\frac{1+\alpha}{2}-\e})$. \end{proof}

The following proposition justifies the depiction in Figure \ref{fig1}.
\begin{proposition}\label{smoothbubble}
Consider the function $\thetafinal$ addressed in the Proposition \ref{Fconvergence}. 
\begin{enumerate}
\item The metric $\thetafinal^4 \geucl$ defines a smooth metric on $S^3$. 
\medbreak
\item The integral of the scalar curvature of $\thetafinal^4 \geucl$ is $16\pi \tfrac{G}{c^2}\int_{\R^3}\omegafinal$.
\end{enumerate}
\end{proposition}

\begin{proof}
Our strategy is to pull back the metric $\thetafinal^4\geucl$ along the inversion with respect to the unit sphere, $x\mapsto \frac{x}{|x|^2}$. Under this inversion the metric $\geucl$ pulls back to $\tfrac{1}{|x|^4}\geucl$ and so to prove our claim it suffices to prove that 
$$\thetafinal\left(\frac{x}{|x|^2}\right)\cdot \frac{1}{|x|}$$
is smooth at $x=0$. In fact, the function is actually analytic near $x=0$. To see this we apply the asymptotic expansion from Proposition \ref{Fconvergence}, assuming that $|x|$ is sufficiently small so that $\tfrac{x}{|x|^2}$ is substantially outside the support of $\omegafinal$. We have 
$$\thetafinal\left(\frac{x}{|x|^2}\right)\cdot \frac{1}{|x|}=\frac{G\mE}{2c^2}+\sum_{l=1}^\infty b_l\cdot C_l\left(\frac{x}{|x|^2}\right)|x|^{2l}=\frac{G\mE}{2c^2}+\sum_{l=1}^\infty b_l\cdot C_l(x).$$
The last equality follows from the fact that $C_l$ is a homogeneous polynomial of degree $l$. This completes the proof of the first of our two claims. The proof of the second claim is a simple computation:
$$\int_{\R^3} R(\thetafinal^4\geucl)\,\dvol_{\thetafinal^4\geucl}=\int_{\R^3} (-8\thetafinal \Delta_{\geucl}\thetafinal)\,\dvol_{\geucl}=\int_{\R^3} 16 \pi \tfrac{G}{c^2}\omegafinal.\qedhere$$
\end{proof}

In the case $\alpha=0$ the region on which the Schwarzschild metric is arising, $\mathrm{J_n}$, vanishes in the limit as a result of Remark \ref{length}. This leaves the limit geometries of $\mathrm{I_n}$ and $\mathrm{V_n}$ which, by inspection of Proposition \ref{bubble}, are $(\R^3,\geucl)$ and $(\R^3,\Theta_F^4\geucl)$ respectively. This reveals a geometric interpretation of the vanishing mass result in that the contents of the bubble are not picked up by a computation of ADM mass. 

\subsection*{The case $\alpha=1$} 

By Corollary \ref{adm1} we see that in this case the ADM mass does not vanish but instead converges to $\mE$. 

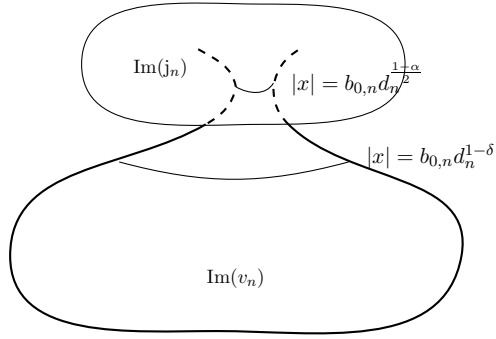
\begin{figure}[h]
\centering
\begin{tikzpicture}[scale=1]

\draw[thick, dashed] (19.5, -4.75) to [out=-30, in=100] (20, -5.25) to [out=-100, in=30] (19.6, -5.75);
\draw[thick] (19.6, -5.75) to [out=-150, in=90] (17, -7.5) to [out=-90, in=180] (19.75, -8.5) to [out=0, in=-100] (23, -7.5) to [out=80, in=-45] (20.7, -5.75);
\draw[thick, dashed] (20.7, -5.75) to [out=135, in=-90] (20.5, -5.25) to [out=80, in=-150] (20.85, -4.75);

\draw[thin] (20, -5.25) to [out=-30, in=-120] (20.5, -5.2);

\draw[thin] (17.95, -4.95) to [out=-90, in=180] (20.35, -5.75) to [out=0, in=-90] (22.25, -4.95) to [out=90, in=0] (20.35, -4.15) to [out=180, in=90] (17.95, -4.95); 

\draw[thin] (18.45, -6.25) to [out=-15, in=-165] (21.5, -6.25);

\node[right] at (20.6, -5.15) {\scalebox{0.75}{$|x|=b_{0,n}d_n^{\frac{1-\alpha}{2}}$}};

\node[right] at (21.6,-6.175) {\scalebox{0.75}{$|x|=b_{0,n}d_n^{1-\delta}$}};

\node at (20, -7.8){\scalebox{0.67}{$\mathrm{Im}(v_n)$}};

\node at (19, -5){\scalebox{0.67}{$\mathrm{Im}(\mathrm{j}_n)$}};

\end{tikzpicture}
\caption{Regions addressed in Proposition \ref{sch}.}\label{fig6}
\end{figure}

The embeddings and domains are 
$$\begin{aligned}
& \mathrm{J_n}=\{b_{0,n}d_n^{1-\delta} \leq |y|\} &&\mathrm{j}_n:y \mapsto y \\
& \mathrm{V_n}=\{|y| \leq b_{0,n}d_n^{-\delta}\} &&\mathrm{v}_n:y \mapsto d_ny. \\
\end{aligned}$$
A direct computation shows that $\mathrm{Im}(\mathrm{J_n})\cup\mathrm{Im}(\mathrm{V_n})=\R^3$.

\begin{proposition}\label{sch}
We have 
$$\begin{aligned}
&\|\mathrm{j}_{n}^*g_n -g_{\mathrm{Schw}}\|_{C^k(\mathrm{J_n}, g_{\mathrm{Schw}})}\to 0,\\
&\left\|d_n^2v_n^*g_n-\Theta_F^4\geucl\right\|_{C^k(V_n,\Theta_F^4\geucl)} \to 0.
\end{aligned}$$
\end{proposition}

\begin{proof}
Taking $\alpha=1$ in Propositions \ref{Jn} and \ref{Vn} completes the proof.
\end{proof}
Note that while we do obtain Schwarzschild initial data in the limit, Remark \ref{bottomsheet} still holds as stated.
 
\subsection*{Summary} 
As a whole, the parameter $\alpha$ is to be understood as a measure of  interaction present in a construction of point particle initial data. The authors in \cite{ZeroMass} assume that the appropriate matter distribution is $\mE \delta$. Withinn our framework their setting corresponds to $\alpha=0$. Proposition \ref{bubble} shows that the effective mass indeed vanishes but that the contents of $\Omega_F$ are trapped in the bubble depicted in Figure \ref{fig1}. 

Our framework also shows that even if one adds back an insufficient amount of matter, corresponding to $0<\alpha<1$, the mass still vanishes in the limit. In the case $\alpha=1$, in conflict to the vanishing mass claim made in \cite{ZeroMass}, we do obtain the initial data $(\R^3\smallsetminus\{0\},g_{\text{Schw}})$. It is only in the case $\alpha=1$ that one is adding  matter at the correct rate to perfectly balance the effects of interaction, thus yielding non-vanishing mass in the limit. 

\subsection*{Locating the horizons}
Proposition \ref{Jn} and Corollary \ref{adm1} suggest the presence of the horizon(s) i.e the outermost minimal surface(s) of $\theta_n^4 \geucl$ near $|x|=b_{0,n} d_n^{\frac{1-\alpha}{2}}$. That this indeed is the case is the content of the following two theorems. Their proofs, however, are somewhat technical and of very different character than the rest of our paper. For this reason we have placed them in Appendix \ref{minsurf-ivaisms}. 

\begin{theorem}\label{minsurf1:thm}
There exist constants $C$ and $N$ such that for all $n\ge N$  the metric $\theta_n^4 \geucl$ has a minimal surface in the region 
$$(1-Cd_n)b_{0,n} d_n^{\frac{1-\alpha}{2}}\le |x|\le (1+Cd_n)b_{0,n} d_n^{\frac{1-\alpha}{2}}.$$
\end{theorem}

To clarify, our definition of outermost minimal surface is that from \cite{Huisken-Ilmanen}. Because of our Theorem \ref{minsurf1:thm},  (the set-up of) Lemma 4.1 from \cite{Huisken-Ilmanen} implies the existence of the outermost minimal surfaces $\Sigma$; we also know that each one of them is a smooth embedded $2$-sphere. 

\begin{theorem}\label{minsurf2:thm}
The outermost minimal surface of $\theta_n^4\geucl$ is connected. Furthermore, there exist constants $C_1, C_2>0$ such that for all $n$ the outermost minimal surface of $\theta_n^4 \geucl$ is located within the region 
$$C_1 b_{0,n} d_n^{\frac{1-\alpha}{2}}\le |x|\le C_2 b_{0,n} d_n^{\frac{1-\alpha}{2}}.$$
\end{theorem}

\section{Connections to the literature}\label{horizon}

\subsection{Relation to the intrinsic flat stability of the Positive Mass Theorem}
Recall that $m_{\mathrm{ADM}}(\theta_n^4\geucl)\to 0$ when $0\le \alpha<1$ due to Corollary \ref{adm1}. In view of the rigidity part of the Positive Mass Theorem \cite{SY79} one might suspect that the manifolds $(\mathbb{R}^3, \theta_n^4\geucl)$ would converge to the Euclidean space in some way. The analysis of our Section \ref{NoahsSection} (compare with Figures \ref{fig1} and \ref{fig2}) proves  that is not literally the case.  A rigorous framework for studying the stability of the rigidity part of the Positive Mass Theorem is proposed in \cite{LeeSormani1}. It has been conjectured that if a sequence of pointed\footnote{More precisely, the conjecture also assumes that $x_n$ do not disappear
down increasingly deep wells.} asymptotically flat manifolds $(M'_n, g_n, x_n)$ with nonnegative scalar curvature whose boundaries are outermost minimal surfaces has $m_{\mathrm{ADM}}(M'_n,g_n)\to 0$, then $(M'_n,g_n)$ converge in the \emph{pointed intrinsic flat sense} to Euclidean space, $({\mathbb{R}^3}, \geucl)$. It is known (see \cite{LeeSormani1}) that the conjecture would be false if it were stated with a stronger notion of convergence (e.g Gromov-Hausdorff convergence).   

The intrinsic flat distance between two oriented Riemannian
manifolds with boundary was originally defined in the joint work of C. Sormani and S. Wenger \cite{SW-JDG}. This distance is measured by first viewing each of the two manifolds as an integral current, pushing forward these integral currents into a common complete metric space via distance preserving maps, and then measuring the flat distance between the two push forwards. To ensure that this notion does not depend upon the choice of particular distance preserving maps, one takes the infimum over all distance preserving maps into all complete metric spaces. 

In practice it is often possible to estimate the intrinsic flat distance by only using notions from Riemannian geometry. A particularly easy-to-use estimate was proven by S. Lakzian and C. Sormani in \cite{LS13}. For the convenience of the reader the full statement of the relevant theorem is included in Appendix \ref{ifl}. 

In the context of our work let $(M'_n,g_n)$ denote that portion of $(\mathbb{R}^3, \theta_n^4\geucl)$ located outside its outermost minimal surface. Intuitively speaking, when $0\le \alpha<1$ Proposition \ref{In} and Theorem \ref{minsurf2:thm} indicate that the sequence $(M'_n, g_n)$ exhausts and converges to the entire Euclidean $\mathbb{R}^3$, as conjectured in \cite{LeeSormani1}. Here is a more precise statement. 

\begin{theorem}\label{ifl:thm}
Fix a point $p\in \mathbb{R}^3\smallsetminus \{ 0\}$ away from the origin, and let $R_0=|p|+1$. 
For all $R>R_0$ the balls $M_{1,n}=B_{g_n}(p,R) \subseteq M_n'$ converge to the Euclidean ball $M_2=B_{\geucl}(p,R) \subseteq \mathbb{R}^3$ in the intrinsic flat sense. 
\end{theorem}

The proof is a direct application of the Lakzian-Sormani estimate and can be found in Appendix \ref{ifl}.

\subsection{Connection to point particle limits of \cite{sb} and \cite{IvasPPpaper}} In  \cite{sb} Gralla and Wald develop a framework for understanding MiSaTaQuWa equations, which are believed to govern the motions of small bodies in general relativity. Their framework involves a one parameter family (``$\e$") of space-times  which satisfies various ``point particle limit" conditions as $\e\to 0$. Examples of initial data which have a potential to produce families of space-times with limit properties of \cite{sb} are constructed in \cite{IvasPPpaper}. In the context of time-symmetric\footnote{For the more general framework please consult \cite{IvasPPpaper}.} initial data ``point particle limits" can be articulated as follows: Let $(M,g)$ be large-scale data, let $S\in M$ and let $(M_0,g_0)$ be asymptotically Euclidean data.  A family of data $(M_\e, g_\e)$ obeys point-particle limit properties with respect to $(M,g)$, $(M_0,g_0)$ and $S\in M$ if the following hold.

\begin{enumerate}
\item  \emph{The ordinary point-particle limit property.} Let $\mathbf{K}\subseteq M\smallsetminus\{S\}$ be a compact set. For small $\e$ there exist embeddings $i_\e:\mathbf{K}\to M_\e$ such that for all $k\in\mathbb{N}\cup\{0\}$ 
$$\|(i_\e)^*g_\e-g\|_{C^{k}(\mathbf{K},g)}\to 0 \text{\ \ as\ \ }\e\to 0.$$

\medbreak
\item \emph{The scaled point-particle limit property.} Let $\mathbf{K}\subseteq M_0$ be a compact set. 
For small $\e$ there exist embeddings $\iota_\e:\mathbf{K}\to M_\e$ such that for all $k\in\mathbb{N}\cup\{0\}$
$$\left\|\tfrac{1}{\e^2}\left(\iota_\e\right)^*g_\e-g_0\right\|_{C^k(\mathbf{K},g_0)}\to 0 \text{\ \ as\ \ }\e\to 0.$$ 
\end{enumerate}

The fact that the family of data $(\R^3,\theta_n^4\geucl)$ satisfies the stated point-particle limit properties is an immediate consequence of Propositions \ref{In} and \ref{Jn}.

\appendix 
\section{Locating the Horizons}\label{minsurf-ivaisms}
This appendix is dedicated to the proofs of Theorems \ref{minsurf1:thm} and \ref{minsurf2:thm}. We begin with the proof of Theorem \ref{minsurf1:thm}, which is more-or-less an Implicit Function Theorem argument applied to the context of Proposition \ref{Jn}. 

\begin{proof}[Proof of Theorem \ref{minsurf1:thm}]
Consider the metric 
$$\psi_n^4\geucl:=\lambda_n^{-2} \mathcal{H}_{\lambda_n}^*(\theta_n^4 \geucl).$$  
As in the proof of Proposition \ref{Jn} we have that on compact subsets of $\mathbb{R}^3\smallsetminus\{0\}$  
\begin{equation}\label{metricapprox}
\psi_n(y) - (1+\tfrac{G\mE}{2c^2|y|}) = O(d_n^{\frac{1+\alpha}{2}}) \text{\ \ as\ \ } n\to \infty;
\end{equation}
furthermore, the same holds for \emph{all} the derivatives. To prove our theorem it suffices to find a minimal surface for $\psi_n^4\geucl$ which is a small perturbation of the sphere $|y|=\frac{G\mE}{2c^2}$. We seek this minimal surface in the form of a (scaled) graph of some positive function $f$ over the unit sphere $S^2$:
$$|y|=\frac{G\mE}{2c^2}f\left(\tfrac{y}{|y|}\right),\ \ \text{i.e.}\ \ y=f(s)\,s\ \  \text{with}\ \ s\in S^2.$$ 
Assuming the standard round metric on $S^2$ throughout, the area element induced on this graph is given by 
$$\left(\frac{G\mE}{2c^2}\right)^2\theta_n(f,s)^4f\sqrt{f^2+|df|^2_{S^2}}\dvol_{S^2}.$$
Ignoring the scalar multiple of $\frac{G\mE}{2c^2}$ the first variation of the area functional is given by 
$$\begin{aligned}
\int \left(4\theta_n^3\tfrac{\partial \theta_n}{\partial f}f\sqrt{f^2+|df|^2}+\theta_n^4\sqrt{f^2+|df|^2}+\theta_n^4\tfrac{f^2}{\sqrt{f^2+|df|^2}}\right)(\delta f)&\dvol_{S^2}\\
-\int \mathrm{div}\left(\theta_n^4\tfrac{f}{\sqrt{f^2+|df|^2}}\, \grad f\right)(\delta f)&\dvol_{S^2}
\end{aligned}$$
Direct expansion of the divergence term, using the decomposition 
$$d\theta_n=\tfrac{\partial \theta_n}{\partial f} df + d_s \theta_n,$$
yields  
$$\begin{aligned}
&4\theta_n^3\left(\tfrac{\partial \theta_n}{\partial f}\tfrac{f|df|^2}{\sqrt{f^2+|df|^2}}+\tfrac{f}{\sqrt{f^2+|df|^2}}\langle d_s\theta_n, df\rangle\right)\\
+&\theta_n^4\left(\tfrac{|df|^2}{\sqrt{f^2+|df|^2}}-\tfrac{f^2|df|^2+f\mathrm{Hess}f(\grad f, \grad f)}{\sqrt{f^2+|df|^2}^3}+\tfrac{f \Delta_{S^2}f}{\sqrt{f^2+|df|^2}}\right).
\end{aligned}$$
Upon an algebraic simplification we obtain the minimal surface equation 
\begin{equation}\label{minsurf:eqn}
\begin{aligned}
\Delta_{S^2}f-\tfrac{1}{f^2+|df|^2}\mathrm{Hess}f(\grad f, \grad f)-\left(2+\tfrac{|df|^2}{f^2+|df|^2}\right)f&\\
-4\theta_n^{-1}\left(\tfrac{\partial \theta_n}{\partial f}f^2 + \tfrac{f}{\sqrt{f^2+|df|^2}}\langle d_s\theta_n, df\rangle\right)&=0.
\end{aligned}
\end{equation}
For the rest of the proof we denote the operator / terms on the first line of \eqref{minsurf:eqn} by 
$\mathcal{P}_0f$.  The approximation \eqref{metricapprox} ensures that 
$$\theta_n^{-1}=(1+\tfrac{1}{f})^{-1}+O\left(d_n^{\frac{1+\alpha}{2}}\right),\ \tfrac{\partial \theta_n}{\partial f}f^2=-1 + O\left(d_n^\frac{1+\alpha}{2}\right),\ d_s\theta_n=O\left(d_n^\frac{1+\alpha}{2}\right)$$
so long as the range of $f$ is contained in a fixed compact subset of $(0,\infty)$. In fact, these estimates hold with first $k$ derivatives so long as there is a uniform bound on the first $k$ derivatives of $f$. 
Overall, this means that \eqref{minsurf:eqn} can be expressed in the form of 
\begin{equation}\label{MinSurf:eqn}
\mathcal{P}_nf:=\mathcal{P}_0f+\langle a_n(f),df\rangle+\Lambda_n(f)=0,
\end{equation}
where 
\begin{itemize}
\item the $1$-form $a_n(f)=a_n(f,s)$, $s\in S^2$ converges to $0$ in the $C^\infty$-sense at the rate of $O\left(d_n^{\frac{1+\alpha}{2}}\right)$. 
\medbreak
\item the function $\Lambda_n(f)=\Lambda_n(f,s)$, $s\in S^2$ converges to the function $4\left(1+\tfrac{1}{f}\right)^{-1}$ in the $C^\infty$-sense at the rate of $O\left(d_n^{\frac{1+\alpha}{2}}\right)$.
\end{itemize}
Note that 
$$\mathcal{P}_\infty f:=\mathcal{P}_0 f+4\left(1+\frac{1}{f}\right)^{-1}=0$$
is the corresponding minimal surface equation for the metric $(1+\tfrac{G\mE}{2c^2|y|})^4\geucl$, and that the constant function $f_\infty=1$ is its solution. To prove our theorem we show that \eqref{MinSurf:eqn} permits a solution for which 
$$\|f-f_\infty\|_{H^k(S^2)}=O(d_n^{\frac{1+\alpha}{2}}) \text{\ \ for\ all\ \ } k.$$

First note that there is a uniform constant $C_P$ such that 
$$\|\mathcal{P}_n(f_\infty)\|_{H^k(S^2)}\le C_P\,d_n^{\frac{1+\alpha}{2}}$$
for all (sufficiently large) $n$.
The linearization $\mathcal{L}_n$ of the operator $\mathcal{P}_n$ at $f_\infty$ takes the form of 
$$\mathcal{L}_n h=\Delta h +\langle b_n, dh\rangle +c_nh$$
where the $1$-form $b_n$ converges to $0$ in the $C^\infty$-sense at the rate of $O(d_n^{\frac{1+\alpha}{2}})$ and where the function $c_n$ converges to the constant function $-1$ in the $C^\infty$-sense at the rate $O(d_n^{\frac{1+\alpha}{2}})$. As such the linearizations $\mathcal{L}_n$ converge to the linearization $\mathcal{L}_\infty=\Delta -1$ of $\mathcal{P}_\infty$ at $f_\infty$:
$$\|\mathcal{L}_n h - \mathcal{L}_\infty h\|_{H^k(S^2)}\le C\, d_n^{\frac{1+\alpha}{2}}\|h\|_{H^{k+1}(S^2)}$$
for some uniform constant $C$. 
Since $\mathcal{L}_\infty:H^{k+2}(S^2)\to H^k(S^2)$ is invertible, we have 
$$\|h\|_{H^{k+2}(S^2)}\le \|\mathcal{L}_\infty h\|_{H^k(S^2)}\le \|\mathcal{L}_n h\|_{H^k(S^2)} + C\, d_n^{\frac{1+\alpha}{2}}\|h\|_{H^{k+1}(S^2)}.$$
For $n$ sufficiently large the last term can be absorbed on the left hand side to yield the uniform invertibility of $\mathcal{L}_n$:
$$\|h\|_{H^{k+2}(S^2)}\le C_L\cdot \|\mathcal{L}_n h\|_{H^k(S^2)};$$
the constant $C_L$ independent of $n$. Inspecting the terms of the remainder 
$$\mathcal{Q}_n(f):=\mathcal{P}_n(f)-\mathcal{P}_n(f_\infty) -\mathcal{L}_n (f-f_\infty)$$ individually we see that for all $\varepsilon>0$ there exists $\nu(\varepsilon)>0$ so that the following holds for all $n\gg1$, $0<\nu<\nu(\varepsilon)$ and $f_1, f_2\in B_\nu(f_\infty)\subseteq H^{k+2}(S^2)$:
$$\|\mathcal{Q}_n(f_1)-\mathcal{Q}_n(f_2)\|_{H^{k}(S^2)}\le \varepsilon \|f_1-f_2\|_{H^{k+2}(S^2)}.$$
In particular, since $\mathcal{Q}_n(f_\infty)=0$, we see that 
$$\|\mathcal{Q}_n(f)\|_{H^{k}(S^2)}\le \varepsilon \nu$$
for all $f\in B_\nu(f_\infty)\subseteq H^{k+2}(S^2)$. 

Choose $\varepsilon$ so that $2C_L\varepsilon<1$ and $n\gg 1$ so that 
$\nu=2C_LC_P d_n^{\frac{1+\alpha}{2}}<\nu(\varepsilon)$.
Then the mapping 
\begin{equation}\label{contraction}
f\mapsto f_\infty-\mathcal{L}_n^{-1}\left(\mathcal{P}_n(f_\infty) +\mathcal{Q}_n(f)\right)
\end{equation}
maps the closed ball $B_{\nu}(f_\infty)\subseteq H^{k+2}(S^2)$ to itself due to 
$$C_LC_Pd_n^{\frac{1+\alpha}{2}}+C_L\varepsilon \nu= C_LC_Pd_n^{\frac{1+\alpha}{2}}(1+2C_L\varepsilon)<\nu.$$ 
Furthermore, the mapping is a contraction since 
$$\begin{aligned}
&\|\mathcal{L}_n^{-1}\left(\mathcal{P}_n(f_\infty) +\mathcal{Q}_n(f_1)\right) - \mathcal{L}_n^{-1}\left(\mathcal{P}_n(f_\infty) +\mathcal{Q}_n(f_2)\right)\|_{H^{k+2}(S^2)}\\
\le &C_L\|\mathcal{Q}(f_1)-\mathcal{Q}(f_2)\|_{H^{k}(S^2)} \le C_L\varepsilon \|f_1-f_2\|_{H^{k+2}(S^2)}\\
\le &\tfrac{1}{2}\|f_1-f_2\|_{H^{k+2}(S^2)}.
\end{aligned}$$
The desired solution $f$ now arises as a fixed point of \eqref{contraction}.
\end{proof}

\bigbreak
Before we proceed to prove Theorem \ref{minsurf2:thm} we remind the reader of several background results. The first result is about a lower bound on the injectivity radius, and it comes out as a consequence of Theorem 4.7 from \cite{CGT} (compare with Theorem 3.7 in \cite{injrad}). For connected, complete Riemannian manifolds with sectional curvature bounds
\begin{equation}\label{sec-curv}
|\mathrm{Sec}(g)| < \kappa
\end{equation}
and for $r<\pi/(4\sqrt{\kappa})$  
we have that 
\begin{equation}\label{lowerbound:injrad}
\mathrm{inj rad}(p)\ge \frac{r}{2}\cdot \frac{\mathrm{Vol}_gB_g(p,r)}{\mathrm{Vol}_gB_g(p,r)+\mathrm{Vol}_{(-\kappa)}(2r)},
\end{equation}
where $\mathrm{Vol}_{(-\kappa)}(2r)$ denotes the volume of the ball of radius $2r$ in the (simply connected) space of constant sectional curvature 
$-\kappa$. 

The second result we review here is  a monotonicity formula for the area of minimal surfaces, e.g. formula (7.5) from \cite{CM}. 
Let $x_0$ be a point on a smooth minimal surface $\Sigma$ in a $3$-manifold with sectional curvature bounds \eqref{sec-curv} and a lower bound $i_0>0$ on the injectivity radius. Then the function  
$$e^{2\sqrt{\kappa}s} s^{-2} \mathrm{Area}_g(B_{g}(x_0,s)\cap \Sigma)$$
of $0<s<\min\{i_0,\tfrac{1}{\sqrt{\kappa}}, \mathrm{dist}_g(x_0, \partial \Sigma)\}$ is non-decreasing. Since the function converges to $\pi$ as $s\to 0$ this monotonicity formula gives us an inequality of the form 
\begin{equation}\label{lowerbound:minsurf}
\mathrm{Area}_g(B_{g}(x_0,s)\cap \Sigma)\geq (\pi e^{-2}) s^2
\end{equation}
on the interval for $s$ stated above. 

The following Proposition is the last remaining background result needed for the proof of Theorem \ref{minsurf2:thm}. The ideas presented here were originally developed for ``Geometrostatic manifolds of small ADM mass" by C. Sormani and I. S., currently in preparation. 

\begin{proposition}\label{christinathing}
Let $\Sigma$ be a minimal surface of $\theta_n^4\geucl$. Suppose $\Sigma$ is diffeomorphic to $S^2$. Then 
$$\pi \le \left(\max_{\Sigma}\left(\theta_n^{-6} |d\theta_n|^2\right)\right) \mathrm{Area}_{\theta_n^4\geucl}(\Sigma).$$
\end{proposition}

\begin{proof}
The mean curvatures of $\Sigma$ computed with respect to two conformally equivalent ambient metrics, $\theta_n^4\geucl$ and $\geucl$, relate as follows:
$$H_{\theta_n^4\geucl}=\theta_n^{-2}H_{\geucl}+4\theta_n^{-3}\grad\theta_n\cdot \vec{N}.$$
Here the gradient, the dot product and the unit normal $\vec{N}$ are all computed with respect to the Euclidean metric $\geucl$. It now follows that $\Sigma$ satisfies the minimal surface equation 
$$H_{\geucl}=-4\theta_n^{-1}\grad(\theta_n)\cdot \vec{N}.$$
The Gauss curvature $K_{\geucl}$ of $\Sigma$ viewed as a submanifold of the Euclidean space satisfies $K_{\geucl}\le \tfrac{1}{4} H_{\geucl}^2$, which is easily seen from the interpretations of said curvatures in terms of the eigenvalues of the shape operator. In particular, we now have 
$$K_{\geucl}\le 4\theta_n^{-2}|d\theta_n|^2.$$ 
The Gauss-Bonnet Theorem implies that 
$$
\begin{aligned}
4\pi =\int_\Sigma K_{\geucl}\dvol \le &4\int_\Sigma \theta_n^{-2}|d\theta_n|^2 \dvol\\
&4\int_\Sigma \left(\theta_n^{-6} |d\theta_n|^2\right) \theta_n^4\dvol\\
\le &4\left(\max_{\Sigma}\left(\theta_n^{-6} |d\theta_n|^2\right)\right) \mathrm{Area}_{\theta_n^4\geucl}(\Sigma)
\end{aligned}
$$
with $\dvol$ on $\Sigma$ referring to the volume element induced by $\geucl$. 
\end{proof}

\begin{proof}[Proof of Theorem \ref{minsurf2:thm}]
In what follows we simplify the notation by setting
$$\varrho_n=b_{0,n}d_n^{\frac{1-\alpha}{2}}.$$
Let $\Sigma$ denote any of the connected components of the outermost minimal surface of $\theta_n^4\geucl$; by \cite{Huisken-Ilmanen} we know that each such $\Sigma$ is diffeomorphic to $S^2$. Our first goal is to estimate $\mathrm{Area}_{\theta_n^4\geucl}(\Sigma)$ using Proposition \ref{christinathing}. We do so by using the series expansion of Proposition \ref{thetaexpands} and by observing that, due to Theorem \ref{minsurf1:thm}, 
\begin{equation}\label{not-too-far}
|x|\ge (1+\mathcal{O}(d_n))\varrho_n
\end{equation}
for all $x\in \Sigma$. Specifically, note that for $x$ satisfying \eqref{not-too-far} we have 
\begin{equation}\label{calc1}
\theta_n^{-3} |d\theta_n|=\frac{\frac{\varrho_n}{|x|^2}+\mathcal{O}(d_n^\alpha)}{\left(1+\frac{\varrho_n}{|x|} +\mathcal{O}(d_n^{\frac{1+\alpha}{2}})\right)^3}=\frac{\varrho_n |x|}{(\varrho_n+|x|)^3}+\mathcal{O}(d_n^\alpha)
\end{equation}
Optimizing the expression $\frac{\varrho_n |x|}{(\varrho_n+|x|)^3}$ over the region \eqref{not-too-far} reveals the maximum of 
$\frac{1}{8\varrho_n}+\mathcal{O}(d_n^\alpha)$
achieved at $|x|=(1+O(d_n))\varrho_n$. 
Proposition \ref{christinathing} now implies $\pi \le \left(\frac{1}{8\varrho_n}+\mathcal{O}(d_n^\alpha)\right)^2 \mathrm{Area}_{\theta_n^4\geucl}(\Sigma)$
i.e. 
\begin{equation}\label{near-penrose}
\begin{aligned}
64\pi\varrho_n^2\left(1+\mathcal{O}(d_n^{\frac{1+\alpha}{2}})\right)^2\le \mathrm{Area}_{\theta_n^4\geucl}(\Sigma).
\end{aligned}
\end{equation}

We now use the Penrose inequality \cite{bray-penrose} to limit the number of connected components of the outermost minimal surface of $\theta_n^4\geucl$. Indeed, if there were a connected component $\Sigma'$ other than the one which contains the minimal surface of Theorem \ref{minsurf1:thm} in its interior, call it $\Sigma$, we would have that \begin{equation}\label{penrose}
\mathrm{Area}_{\theta_n^4\geucl}(\Sigma)+ \mathrm{Area}_{\theta_n^4\geucl}(\Sigma')\le 16\pi \left(\frac{G}{c^2}m_{\mathrm{ADM}}(\theta_n^4\geucl)\right)^2=64\pi \varrho_n^2. 
\end{equation}
For $n$ sufficiently large \eqref{penrose} is in contradiction with \eqref{near-penrose}. In the rest of the proof $\Sigma_n$ denotes the outermost minimal surface which contains the minimal surface of Theorem \ref{minsurf1:thm} in its interior.

Our next observation is that for all $C>1$ (and sufficiently large $n$) the surface $\Sigma_n$ must contain a point in $B_{\geucl}(0,C\varrho_n)$. To prove this, fix $C$ and suppose the opposite: that 
$$|x|\ge C\varrho_n$$
for all $x\in \Sigma_n$. Under this assumption we obtain (compare to \eqref{calc1}) that
$$\theta_n^{-3}|d\theta_n|\le \frac{1}{\varrho_n}\cdot \frac{C}{(1+C)^3}+\mathcal{O}(d_n^\alpha).$$
From the Penrose Inequality $\mathrm{Area}_{\theta_n^4\geucl}(\Sigma_n)\le 64\pi \varrho_n^2$ and Proposition \ref{christinathing} we see that 
$$\pi \le \left(\frac{1}{\varrho_n}\cdot \frac{C}{(1+C)^3}+\mathcal{O}(d_n^\alpha)\right)^2\cdot 64\pi \varrho_n^2.$$
In particular, we arrive at 
$$\frac{1}{8}\le \frac{C}{(1+C)^3}+\mathcal{O}(d_n^\alpha),$$
which is impossible because $C>1$ implies $\frac{C}{(1+C)^3}<\frac{1}{8}$.
This contradiction shows that 
$$\Sigma_n \cap B_{\geucl}(0,C\varrho_n)\neq \emptyset \text{\ \ for all\ \ } C>1.$$
The value of $C_2=2C$ which completes the proof of our theorem is described later on in the proof.  

Suppose now that $\Sigma_n$ contains a point outside of $B_{\geucl}(0,2C\varrho_n)$, and consider the surface 
$$\Sigma'_n= \Sigma_n \cap \left(\bar B_{\geucl}(0,2C\varrho_n)\smallsetminus B_{\geucl}(0, C\varrho_n)\right).$$
Our next step is to estimate the area of $\Sigma'_n$ from below using \eqref{lowerbound:minsurf}. To do so we apply \eqref{lowerbound:injrad} to the complete metric 
$$\tilde{\theta}_n^4\geucl:=\left(\chi_n\theta_n+(1-\chi_n)\right)^4\geucl$$ which is designed to replace $\theta_n^4\geucl$ in regions where its curvature is too high while preserving $\theta_n^4\geucl$ near $\Sigma'_n$. Specifically, we take $\chi_n(x)=\chi(x/(C\varrho_n))$
to be a self-similar family of cut-off functions with 
$$\begin{cases}
\chi_n(x)\equiv 1, \text{\ \ if\ \ } |x|\ge C\varrho_n/2,\\
\chi_n(x)\equiv 0, \text{\ \ if\ \ } |x|\le C\varrho_n/4,\\
\partial \chi_n=O((C\varrho_n)^{-1}), \\
\partial^2 \chi_n=O((C\varrho_n)^{-2}).  
\end{cases}$$

The series expansion of $\theta_n$ from Proposition \ref{thetaexpands} shows that 
$$|\theta_n|=1+O(C^{-1}),\ \ |\partial \theta_n|=O((C^2\varrho_n)^{-1}) \text{\ \ and\ \ }
|\partial^2 \theta_n|=O((C^3\varrho_n^2)^{-1})$$ 
on the regions where $\chi_n\neq 0$. The choice of our $\chi_n$ now ensures that 
$$\left|\partial \tilde{\theta}_n\right|=O((C^2\varrho_n)^{-1}) \text{\ \ and\ \ } \left|\partial^2 \tilde{\theta}_n\right|=O((C^3\varrho_n^2)^{-1});$$
consequently, we obtain 
$$\left|\mathrm{Sec}\left(\tilde{\theta}_n^4\geucl\right)\right|=O(\kappa) \text{\ \ with\ \ } \kappa=\frac{1}{C^4\varrho_n^2}.$$
Next, note that 
$B_{\tilde{\theta}_n^4\geucl}(p,r)\subseteq B_{\geucl}(p,r)$ due to $\theta_n>1$. It follows that 
$$\mathrm{Vol}_{\tilde{\theta}_n^4\geucl}B_{\tilde{\theta}_n^4\geucl}(p,r)\ge \tfrac{4\pi}{3 }r^3 \text{\ \ and\ \ }\mathrm{inj rad}(p)\ge \frac{r}{2}\cdot \frac{\frac{4\pi}{3 }r^3}{\frac{4\pi}{3 }r^3+\mathrm{Vol}_{(-\kappa)}(2r)}.$$
By taking $r=\frac{\pi}{8\sqrt{\kappa}}$, for example, we get 
$\mathrm{Vol}_{(-\kappa)}(2r)=\pi \frac{\sinh(\pi/2)-\pi/2}{\sqrt{\kappa}^3}$, and 
$$\mathrm{inj rad}\ge \frac{C_{\mathrm{inj}}}{\sqrt{\kappa}}=C_{\mathrm{inj}}C^2\varrho_n$$
for some (universal) constant $C_{\mathrm{inj}}$.

Choose a point $x_0\in \Sigma'_n$ such that $|x_0|=\tfrac{3}{2}C\varrho_n$. Observe that  
$$\tfrac{1}{2}C\varrho_n\le \mathrm{dist}_{\tilde{\theta}_n^4\geucl}(x_0,\partial \Sigma'_n)\le \mathrm{Const}\cdot C\varrho_n$$
due to the fact that  
$\tilde{\theta}_n=\theta_n=1+O(C^{-1})$ on the region of interest.
Assuming that $C\gg \tfrac{1}{C_{\mathrm{inj}}}$, we have 
$$\tfrac{1}{2}C\varrho_n< \mathrm{min}\{\mathrm{inj rad}, \tfrac{1}{\sqrt{\kappa}}, \mathrm{dist}_{\tilde{\theta}_n^4\geucl}(x_0,\partial \Sigma'_n)\}.$$
In particular, we may apply \eqref{lowerbound:minsurf} with $s=\tfrac{1}{4}C\varrho_n$:
$$\mathrm{Area}_{\theta_n^4\geucl}(\Sigma_n)\ge \mathrm{Area}_{\tilde{\theta}_n^4\geucl}(B_{\tilde{\theta}_n^4\geucl}(x_0,s)\cap \Sigma'_n)\ge (C^2\pi e^{-2}/16)\varrho_n^2$$
For a (universal) large constant $C$ the last inequality contradicts the Penrose inequality. Our proof is now complete. 
\end{proof}

\section{The Intrinsic Flat Distance}\label{ifl}

Our proof of Theorem \ref{ifl:thm} relies on an estimate for the intrinsic flat distance extracted from \cite{LS13}.

\begin{theorem} \label{thm-subdiffeo} 
Suppose $(M_1,g_1)$ and $(M_2,g_2)$ are oriented
precompact Riemannian manifolds with diffeomorphic subregions $W_i \subseteq M_i$.
Identifying $W_1=W_2=W$ assume that
on $W$ we have
\begin{equation*}
g_1 \le (1+\varepsilon)^2 g_2 \textrm{ and }
g_2 \le (1+\varepsilon)^2 g_1
\end{equation*}
Let $\mathrm{Diam}= \max\{1,\mathrm{diam}(M_1), \mathrm{diam}(M_2)\}$, $\lambda=\displaystyle{\sup_{q_1,q_2 \in W}} |d_{M_1}(q_1,q_2)-d_{M_2}(q_1,q_2)|$ and 
$$a>\frac{\arccos(1+\varepsilon)^{-1} }{\pi}\mathrm{Diam},\ \ 
\bar{h}= \max\{\sqrt{2\lambda \mathrm{Diam} },  \sqrt{\varepsilon^2 + 2\varepsilon} \; \mathrm{Diam} \}.$$
Then
\begin{align*}
d_{\mathcal{IF}}(M_1, M_2) \le&
\left(2\bar{h} + a\right) \Big(
\mathrm{Vol}_{g_1}(W)+\mathrm{Vol}_{g_2}(W)+\mathrm{Vol}_{g_1}(\partial W)+\mathrm{Vol}_{g_2}(\partial W)\Big)\\
&+\mathrm{Vol}_{g_1}(M_1\setminus W)+\mathrm{Vol}_{g_2}(M_2\setminus W).
\end {align*}
\end{theorem}

\begin{proof}[Proof of Theorem \ref{ifl:thm}]
Fix $0<\e<\tfrac{1-\alpha}{4}$ and consider 
$$U_n=B_{\geucl}(p,R)\smallsetminus \{|x|\ge b_{0,n}d_n^{\frac{1-\alpha}{2}-\e}\}.$$ 
By (the proof of) Proposition \ref{In} we have 
$$\geucl\le g_n\le (1+\mathcal{O}(d_n^\e))\geucl \text{\ \ on\ \ } U_n.$$
Now let $q_1, q_2\in U_n$. A consideration of the straight line segment possibly interrupted by a semi-circular arc shows that    
$$|q_1-q_2|\le d_{(U_n, g_n)}(q_1, q_2)\le (1+\mathcal{O}(d_n^\e))\left(|q_1-q_2|+\pi b_{0,n}d_n^{\frac{1-\alpha}{2}-\e}\right).$$
Define $\tilde{\lambda}_n:=\displaystyle{\sup_{q_1,q_2 \in U_n}} |d_{U_n}(q_1,q_2)-d_{\mathbb{R}^3}(q_1,q_2)|$. By our choice of $\e$ we have 
$$\tilde{\lambda}_n=\mathcal{O}\left(Rd_n^\e\right).$$

Let $W_n=B_{\geucl}(p,R-\tilde{\lambda}_n)\smallsetminus \{|x|\ge b_{0,n}d_n^{\frac{1-\alpha}{2}-\e}\}\subseteq U_n$ and note that  
$$W_n\subseteq M_{1,n}\subseteq M_2.$$
With this choice of $W_n$ the value of $\lambda$ needed to apply Theorem \ref{ifl:thm} can still be taken to be $\tilde{\lambda}_n$. Consequently, may choose $a$ and $\bar{h}$ with 
$$a=\mathcal{O}(Rd_n^{\frac{\e}{2}}),\ \ \bar{h}=\mathcal{O}(Rd_n^{\frac{\e}{2}}).$$
In addition, Theorem \ref{minsurf2:thm} implies  
$$\begin{aligned}
M_{1,n}\smallsetminus W_n&\subseteq \{C_1b_{0,n}d_n^{\frac{1-\alpha}{2}}\le |x|\le b_{0,n}d_n^{\frac{1-\alpha}{2}-\e}\}\cup \{R-\tilde{\lambda}_n\le |x-p|\le R\}\\
M_2\smallsetminus W_n&\subseteq \{|x|\le b_{0,n}d_n^{\frac{1-\alpha}{2}-\e}\}\cup \{R-\tilde{\lambda}_n\le |x-p|\le R\}.
\end{aligned}$$
By Proposition \ref{thetaexpands} we have $\theta_n=\mathcal{O}(1)$ for $|x|\ge C_1b_{0,n}d_n^{\frac{1-\alpha}{2}}$ and thus
$$
d_{\mathcal{IF}}(M_{1,n}, M_2) \leq
O(Rd_n^{\frac{\e}{2}})\cdot (\mathcal{O}(R^3)+\mathcal{O}(R^2))
 +\mathcal{O}(d_n^{3(\frac{1-\alpha}{2}-\e)}) + \mathcal{O}(\tilde{\lambda}_nR^2),
$$
which can clearly be made as small as possible. 
\end{proof}

\bibliographystyle{plain}
\bibliography{2014}

\end{document}